%% file: grcontracts.tex
\documentclass[preprint,3p,number]{elsarticle}

\usepackage{JLAMP}

\begin{document}

\begin{frontmatter}

\title{Unifying Theories of Reactive Design Contracts}

\author{Simon Foster, Ana Cavalcanti, Samuel Canham, Jim Woodcock and Frank Zeyda}

\begin{abstract}
  Design-by-contract is an important technique for model-based design in which a composite system is specified by a
  collection of contracts that specify the behavioural assumptions and guarantees of each component. In this paper, we
  describe a unifying theory for \emph{reactive design} contracts that provides the basis for modelling and verification
  of reactive systems. We provide a language for expression and composition of contracts that is supported by a rich
  calculational theory. In contrast with other semantic models in the literature, our theory of contracts allow us to
  specify both the evolution of state variables and the permissible interactions with the environment. Moreover, our
  model of interaction is abstract, and supports, for instance, discrete time, continuous time, and hybrid computational
  models. Being based in Unifying Theories of Programming (UTP), our theory can be composed with further computational
  theories to support semantics for multi-paradigm languages. Practical reasoning support is provided via our proof
  framework, Isabelle/UTP, including a proof tactic that reduces a conjecture about a reactive program to three
  predicates, symbolically characterising its assumptions and guarantees about intermediate and final observations. This
  allows us to verify programs with a large or infinite state space. Our work advances the state-of-the-art in semantics
  for reactive languages, description of their contractual specifications, and compositional verification.
\end{abstract}

\end{frontmatter}

\section{Introduction}
\label{sec:intro}
\input{sec/intro.tex}

\section{Preliminaries}
\label{sec:prelim}
\input{sec/prelim.tex}

\section{Reactive Relations and Conditions}
\label{sec:rrel}
\input{sec/rrel.tex}

\section{Reactive Design Contracts}
\label{sec:contracts}
\input{sec/contracts.tex}

\section{Automated Verification using Contracts}
\label{sec:mondex}
\input{sec/mondex.tex}

\section{Theory of Generalised Reactive Designs}
\label{sec:grd}
\input{sec/grd.tex}

\section{Mechanised Proof with Isabelle/UTP}
\label{sec:mech}
\input{sec/mech.tex}

\section{Related Work}
\label{sec:related}
\input{sec/related.tex}

\section{Conclusions}
\label{sec:conclusion}
\input{sec/conclusions.tex}

\section*{Acknowledgements}

This research is funded by the CyPhyAssure project\footnote{CyPhyAssure Project:
  \url{https://www.cs.york.ac.uk/circus/CyPhyAssure/}}, EPSRC grant EP/S001190/1, the RoboCalc project\footnote{RoboCalc
  Project: \url{https://www.cs.york.ac.uk/circus/RoboCalc/}}, EPSRC grant EP/M025756/1, and the EU Horizon 2020 project
``INTO-CPS'', grant agreement 644047~\footnote{INTO-CPS Project Website: \url{http://projects.au.dk/into-cps/}.}. We
would like to thank our colleagues on this project, particularly those from Peter Gorm Larsen's group at {\AA}rhus
University, for their collaboration over the past three years, without which this work could not have come to
fruition. We also thank the anonymous reviewers of this article, whose suggestions have greatly improved the
presentation of our work.

\bibliographystyle{spmpsci}      %
\bibliography{grcontracts}   %

\end{document}

%% file: sec/intro.tex
Verification of large-scale systems of systems and cyber-physical systems is challenging due to the size and complexity
of the underlying models~\cite{Vincentelli2012}. The design-by-contract paradigm~\cite{Meyer92,Jones2003,Bauer2012}
provides a precise approach for compositional verification. In this approach, the verifier defines contracts with two
parts: (1) the required behaviour that a constituent guarantees to implement, and (2) the assumptions the constituent
can make about its environment~\cite{Benveniste2007}.

A composite system can thus be characterised by a number of contracts, one for each constituent, the composition of
which fulfils the overall system-level contract that specifies the behaviour of the system as a whole. Each constituent
system can be shown to fulfil its required behaviour under certain assumptions. Violation of a constituent's assumptions
leads to unpredictable behaviour of the entire system. To enable verification, we also need a theory of contracts
applicable to a wide range of paradigms and accompanied by automated tool support.

In this article, we provide a novel unifying theory of contracts for a wide-spectrum of stateful reactive
languages~\cite{Oliveira&09,Woodcock14,Foster16b}, with practical verification support provided in the Isabelle/HOL
theorem prover~\cite{Isabelle,Foster16a}. Our contracts are supported by a rich algebraic theory with composition
operators and laws to calculate the overall assumptions and guarantees of composite contracts. Whilst contracts can be
applied as a specification mechanism for the purpose of verification, they can also function as a denotational model for
reactive languages according to the ``programs-as-predicates'' philosophy~\cite{Hehner93,Oliveira&09}. This means that
our contract theory is a semantic model sufficient to characterise both specifications and implementations, and thus
avoids a formalisation gap between different languages.

Our theory is based in Hoare and He's Unifying Theories of Programming~\cite{Hoare&98,Cavalcanti&06} (UTP): a meta-model
framework for describing denotational semantics in terms of an alphabetised predicate calculus that acts as a
\emph{lingua franca}. Different semantic models can be encoded into this common domain, and compared and linked with one
another. Unlike other works~\cite{He2005,Zhan2008}, our theory of contracts is wholly embedded into the alphabetised
relational calculus: this gives a direct route to automated reasoning. We build on previous UTP theories of reactive
processes~\cite{Hoare&98} and reactive designs~\cite{Cavalcanti&06,Oliveira&09}, whilst making several improvements and
generalisations.

Our UTP theory is highly extensible in several ways. Whilst the previous work supports only discrete sequence-based
traces, we adopt our work on algebraic trace models~\cite{Foster17a}. This allows us to describe contracts for a
spectrum of languages, from untimed and discrete time languages, through to continuous time and hybrid systems with
differential equations. Moreover, our theory can be extended with additional semantic information, such as refusals as
employed in \Circus~\cite{Woodcock2001-Circus,Oliveira&09}, but also potentially other models such as timed testing
traces~\cite{Woodcock14,Foster16b}.

Finally, our theory has been mechanised in our Isabelle-based proof assistant for the UTP, which we call
Isabelle/UTP~\cite{Foster16a,Zeyda16,Foster14}. This provides us with theory construction and verification facilities,
and also the ability to develop verification tools from UTP-based semantic models. Notably, we demonstrate a prototype
proof tactic for proving refinements between reactive contracts by re-expressing refinement conjectures as three
implications between the pre-, peri-, and postconditions that are expressed purely in relational calculus. These proof
obligations can then be discharged using relational and predicate calculus tactics in Isabelle/HOL, which greatly
improves the potential for automation. Moreover, our contract theory allows us to characterise and reason about reactive
programs purely symbolically, which allows us to verify programs with a very large or infinite state.

All the theorems proved here are proved in Isabelle/UTP; these proofs can be found in our Isabelle/UTP
repository\footnote{Isabelle/UTP Repository: \url{https://github.com/isabelle-utp/utp-main}}. Additionally, most
theorems and definitions in the paper are accompanied by a small Isabelle icon (\isalogo). In the electronic version,
each icon is hyperlinked to the corresponding mechanised artefact in the repository. This, we hope, will convince the
reader of the level of rigour employed in this work.

\vspace{.5ex}

\noindent In summary, the novel contributions of this paper are as follows:
\begin{enumerate}
  \item a novel UTP theory of stateful reactive contracts based on a generalised semantic trace model;
  \item a contract notation with assumptions, and guarantees of intermediate and final states;
  \item theorems for calculating composite contracts, including parallel composition;
  \item an automated proof method for proving contractual refinement conjectures.
\end{enumerate}

\noindent The remainder of the paper is organised as follows. In Section~\ref{sec:prelim}, we provide the necessary
context for our work, including details of UTP, reactive processes, and reactive designs. In Section~\ref{sec:rrel}, we
describe our theories of reactive relations and conditions, which form an important building block of our contracts. In
Section~\ref{sec:contracts}, we introduce reactive design contracts; including their notation, algebraic laws, and a
number of small illustrative examples. In Section~\ref{sec:mondex}, we use the laws described in
Section~\ref{sec:contracts} to demonstrate the use of contracts in reasoning about a small case study, which has also
been mechanised in Isabelle/UTP. In Section~\ref{sec:grd}, we provide a detailed overview of our UTP theory, including
its healthiness conditions, signature, and justification for the laws presented in Section~\ref{sec:contracts}. In
particular we show how recursion and parallel composition are handled. In Section~\ref{sec:mech}, we give an overview of
our mechanisation of reactive designs in Isabelle, highlight idiosyncrasies needed for the encoding, and describe our
reactive-design proof tactics. In Section~\ref{sec:related}, we give a survey of related work on design-by-contract and
explain the distinguishing features of our work. Finally, in Section~\ref{sec:conclusion} we conclude.

%% file: sec/prelim.tex
In this section we review preliminaries of UTP and its core theories, and introduce a number of foundational theorems,
including novel results regarding the class of continuous healthiness conditions (Theorem~\ref{thm:contthm}). All
theorems in this and following sections have been mechanically proved in Isabelle/UTP~\cite{Foster16a}; for details
please see Section~\ref{sec:mech} and the accompanying Isabelle proofs. This, we believe, adds a substantially more
rigorous basis for UTP and the presented results than the previous works.

\subsection{UTP}
\label{sec:utp}

UTP~\cite{Hoare&98} seeks to identify the fundamental computational paradigms that exist as foundations of programming
language semantics and formalises them using UTP theories. UTP theories can describe what it means for a language to be
concurrent~\cite{Hoare&98,Cavalcanti&06}, real-time~\cite{Sherif2010}, or object-oriented~\cite{SCS06}. The UTP thus
promotes reuse of theoretical building blocks that underlie programming languages. 

UTP is based on an alphabetised relational calculus~\cite{Hoare&98} with operators of higher-order predicate calculus
and relation algebra. An alphabetised predicate $P$ consists of a set of typed variables $\alpha(P)$, and a predicate
which may refer to only those variables in $\alpha(P)$. Alphabetised relations are alphabetised predicates whose
alphabet consists of pairs of variables that denote initial values ($x$) and later values ($x'$). The relations are
ordered by refinement, $P \refinedby Q$, which is denoted below:

\begin{definition} $P \refinedby Q ~ \defs ~ [Q \implies P]$ \label{def:refinement} where
  $\alpha(P) = \alpha(Q)$ \isalink{https://github.com/isabelle-utp/utp-main/blob/07cb0c256a90bc347289b5f5d202781b536fc640/utp/utp_pred.thy\#L141}
\end{definition}

\noindent Here, $[P]$ denotes the universal closure of $P$; if $\alpha(P) = \{x_1 \cdots x_n\}$, then
$[P] \defs (\forall x_1 \cdots x_n @ P)$. Refinement states that the predicate $Q$ implies $P$ over all variables, and
thus $P$ sets an upper bound on the possible observations exhibited by $Q$. Consequently, $\refinedby$ is a partial
order on the set of alphabetised relations.

The domain of alphabetised relations forms a number of important algebraic structures including (1) a complete
lattice~\cite{Hoare&98}, where the order is refinement ($\refinedby$), $\false$ is top, and $\true$ is bottom; (2) a
relation algebra~\cite{Tarski41,Foster14}; (3) a cylindric algebra~\cite{Tarski71,Foster16a}; and (4) a
quantale~\cite{Foster16a}, which induces (5) a Kleene algebra~\cite{Armstrong2015}. Together these provide a rich set of
base properties supporting program verification~\cite{Armstrong2015,Foster16a}. 

UTP follows the ``programs-as-predicates'' approach~\cite{Hehner93}, where a program is modelled as a relation in the
alphabetised predicate calculus. For example, an assignment $x := x + 1$, which increments variable $x$ whilst leaving
all other variables unchanged, can be denoted using the predicate $x' = x + 1 \land y' = y$, where $y$ denotes the
collection of variables other than $x$. Programming operators, such as sequential composition ($P \relsemi Q$),
if-then-else conditional ($\conditional{P}{b}{Q}$), non-deterministic choice ($P \sqcap Q$), and recursion ($\mu F$),
are denoted as predicates~\cite{Hoare&98,Cavalcanti04}, as shown below.

\begin{definition}[UTP Programming Operators] \label{def:utp-prog} \isalink{https://github.com/isabelle-utp/utp-main/blob/07cb0c256a90bc347289b5f5d202781b536fc640/utp/utp_rel_laws.thy}
  \begin{align*}
    P \relsemi Q &\defs \exists v_0 @ P[v_0/v'] \land Q [v_0/v] & \alpha(P) = \alpha(Q) = \{v, v'\} \\
    x := v &\defs x' = v \land y' = y \land \cdots \land z' = z \\
    \II    &\defs x := x \\
    \bigsqcap_{i \in I} \, P(i)  &\defs \bigvee i \in I @ P(i) \\
    \conditional{P}{b}{Q} &\defs (b \land P) \lor (\neg b \land Q) & \alpha(P) = \{v, v'\}, \alpha(b) = \{v\} \\
    \mu F &\defs \bigsqcap \{ X | F(X) \refinedby X \}
  \end{align*}
\end{definition}

Moreover, the empty relation $\false$ is usually used to denote either a program that fails to terminate, or else is
``miraculous'', having no possible behaviours. This embedding of programs into logic naturally provides great
opportunities for verification by automated proof~\cite{Foster16a}. Moreover, the standard laws of
programming~\cite{Hoare87} are all theorems with respect to the operator denotations. We also emphasise that these
operators are all alphabet polymorphic, and can therefore be used to compose predicates of varying types, so long as the
side conditions are satisfied. A selection of theorems of Definition~\ref{def:utp-prog} is shown below.

\begin{theorem}[Relational Calculus Laws] \label{thm:rellaws} \isalink{https://github.com/isabelle-utp/utp-main/blob/90ec1d65d63e91a69fbfeeafe69bd7d67f753a47/utp/utp_rel_laws.thy}

\noindent\begin{minipage}{.5\linewidth}
\begin{align*}
  P \relsemi \II = \II \relsemi P &= P \\
  P \relsemi \false = \false \relsemi P &= \false \\
  (P \relsemi Q) \relsemi R &= P \relsemi (Q \relsemi R) \\
  (\conditional{P}{b}{Q}) \relsemi R &= \conditional{P \relsemi R}{b}{Q \relsemi R}
\end{align*}
\end{minipage}
\begin{minipage}{.5\linewidth}
\begin{align*}
  \left(\bigsqcap_{i \in I} \, P(i)\right) \relsemi Q &= \bigsqcap_{i \in I} \, P(i) \relsemi Q \\[1ex]
  P \relsemi \left(\bigsqcap_{i \in I} \, Q(i)\right) &= \bigsqcap_{i \in I} \, P \relsemi Q(i)
\end{align*}
\end{minipage}
\end{theorem}

The relational skip ($\II$) is a left and right identity, and $\false$ is a left and right annihilator. Sequential
composition is associative, it distributes through conditional from the right (but not the left), and it distributes
through internal choice from the left and right.

UTP theories are well-defined subsets of the alphabetised relations that satisfy certain properties desirable for a
particular computational paradigm. For example, to model real-time programs, we need a way of recording how much time
has passed since execution began, and ensuring that the passage of time is well-behaved by, for instance, forbidding
reverse time travel. In UTP, this can be achieved by extending our alphabet with special observational variables, such
as $clock, clock' : \nat$, which can be used to model time, and imposing invariants, such as $clock \le clock'$. UTP
theories are therefore specified in terms of three parts:
\begin{enumerate}
\item an alphabet of typed \emph{observational variables}, which are used to encode observable semantic quantities
  important for the theory;
\item a \emph{healthiness condition} ($\healthy{HC}$), that specifies the invariants as a function from predicates to
  predicates with the above alphabet;
\item a \emph{signature}: that is, a set of constructors and operators that are healthy elements of the theory.
\end{enumerate}
A theory's alphabet is often open to extension, such that additional observational variables can be added, or
the types of variables specialised, assuming a notion of subtyping exists. This also means that UTP theories can readily
be combined by merging the alphabets and composing the healthiness conditions.

If a relation is a fixed-point of the healthiness conditions, $P = \healthy{HC}(P)$, then it is said to be
$\healthy{HC}$-healthy. For our theory of real-time, we can specify a healthiness condition
$\healthy{HT}(P) \defs (P \land clock \le clock')$, which conjoins a relation with the invariant. Then we have it that
$$\textsf{delay}(n : \nat) \defs clock' = clock + n$$ is a $\healthy{HT}$-healthy relation, since it always satisfies $clock \le clock'$, and therefore is
is a fixed-point of $\healthy{HT}$.

A UTP theory's domain is the set of healthy predicates: $\theoryset{\healthy{HC}} \defs \{P | \healthy{HC}(P) = P\}$.
For this reason, it is necessary that a healthiness function is idempotent
($\healthy{HC} \circ \healthy{HC} = \healthy{HC}$), and usually also monotonic. Monotonicity ensures that the UTP theory
forms a complete lattice, substantiated by the Knaster-Tarski theorem~\cite{Tarski55}. This gives rise to a theory top
($\thtop{T}$), a bottom ($\thbot{T}$), an infimum ($\thinf{T}\,A$), for $A \subseteq \theoryset{\healthy{HC}}$, a
supremum ($\thsup{T}\,A$), a least fixed-point operator $\thlfp{T}\,F$, for
$F : \theoryset{\healthy{HC}} \to \theoryset{\healthy{HC}}$ and a greatest fixed-point operator $\thgfp{T}\,F$. The top and
bottom can be obtained by applying the healthiness condition to $\false$ and $\true$, respectively. However, the induced
lattice does not in general share the same operators as the alphabetised predicate lattice. Thus, for our purposes, we
are interested in the stronger property of \textbf{continuity}, which gives rise to additional properties.
\begin{definition}[Continuous Healthiness Conditions] \isalink{https://github.com/isabelle-utp/utp-main/blob/07cb0c256a90bc347289b5f5d202781b536fc640/utp/utp_healthy.thy\#L112}

  \noindent\healthy{HC} is said to be continuous if it satisfies $\healthy{HC}(\bigsqcap A) = \bigsqcap \{ \healthy{HC}(P) | P \in A
  \}$ for $A \neq \emptyset$.
\end{definition}
\noindent This notion of continuity, also know as universal disjunctivity, is stronger than the related notion of
Scott-continuity~\cite{Scott71}, which requires that $A$ also be directed Every continuous healthiness condition is also
monotonic and thus induces a complete lattice. Continuity also means that the theory's infimum ($\thinf{T}$) is the same
operator as the alphabetised predicate infimum ($\bigsqcap$) for non-empty sets. So, a number of additional laws can be
imported into the theory, some of which are illustrated below.
\begin{theorem}[Continuous Theory Laws] \label{thm:contthm} \isalink{https://github.com/isabelle-utp/utp-main/blob/07cb0c256a90bc347289b5f5d202781b536fc640/utp/utp_theory.thy}
\begin{align}
  \thbot{T} &~~=~~ \healthy{HC}(\true) \\
  \thtop{T} &~~=~~ \healthy{HC}(\false) \label{thm:utp-lattice-top} \\
  \thbot{T} \intchoice P &~~=~~ \thbot{T} & \text{if $P$ is $\healthy{HC}$-healthy} \\
  \thtop{T} \intchoice P &~~=~~ P &\text{if $P$ is $\healthy{HC}$-healthy} \\
  \thinf{T} \, A &~~=~~ \bigsqcap \, A &\text{if $A \neq \emptyset$ and $A \subseteq \theoryset{\healthy{HC}}$} \\
  \thlfp{T} X @ F(X) & ~~=~~ \mu X @ F(\healthy{HC}(X)) & \text{if}~~ F : \theoryset{\healthy{HC}} \to \theoryset{\healthy{HC}} \label{thm:contthm-lfp}
\end{align}
\end{theorem}
For the first four these identities, monotonicity of $\healthy{HC}$ is actually a sufficient assumption. Of particular
interest is \eqref{thm:contthm-lfp} that shows how a theory's weakest fixed-point operator ($\thlfp{T}$) can be
rewritten to the alphabetised-predicate weakest fixed-point ($\mu$). The requirement is that the continuous healthiness
condition $\healthy{HC}$ can be applied after each unfolding of the fixed-point to ensure that the function $F$ is only
ever presented with a healthy predicate. It should be noted that this requirement for continuous healthiness functions
does not similarly restrict fixed-point functions. Specifically, our laws apply for any monotonic function $F$ and thus
at this level there is no restriction in modelling of unbounded non-determinism. Indeed, we have encountered very few
healthiness conditions that are monotonic but not continuous.

Healthiness conditions in UTP are often built by composition of several component functions. That being the case,
continuity and idempotence properties of the overall healthiness condition can also be obtained by composition.

Though UTP was originally a purely theoretical framework for denotational semantics~\cite{Hoare&98}, more recently it
has been adapted into an implementation within the Isabelle proof assistant~\cite{Isabelle}, called
Isabelle/UTP~\cite{Foster14,Foster16a}. This proof tool can be used develop UTP theories, by defining healthiness
conditions and proving algebraic laws, in order to support mechanically verified denotational semantics. Moreover, such
a mechanised denotational semantics can be used to construction verification tools for different languages by harnessing
Isabelle's powerful automated proof facilities~\cite{Blanchette2011}. In this article, we use Isabelle/UTP to both
define and verify our UTP theory of reactive contracts, and to produce an associated verification technique.

\subsection{Designs}
\label{sec:designs}

The UTP theory of designs~\cite{Hoare&98,Cavalcanti04} has two observational variables, $ok, ok' : \Bool$, flags that
denote whether a program was started and whether it terminated, respectively. The design, $\design{P_1}{P_2}$, states
that if a program is started and the state satisfies precondition $P_1$, then it will terminate and satisfy
postcondition $P_2$. This is encoded in the following predicative definition.
\begin{definition} \label{def:design} $\design{P_1}{P_2}   ~~\defs~~ (ok \land P_1) \implies (ok' \land P_2)$ \isalink{https://github.com/isabelle-utp/utp-main/blob/07cb0c256a90bc347289b5f5d202781b536fc640/theories/designs/utp_des_core.thy\#L111}
\end{definition}
\noindent Here, $P_1$ and $P_2$ are relations on variables excluding $ok$ and $ok'$. Effectively, this encoding allows a
pair of predicates to be encoded as a single predicate. An simple example is the design
$D \defs \design{\true}{x' = 1}$, where $x : \nat$ is a program variable. If $D$ is given permission to execute
($ok = true$), then the program terminates ($ok' = true$) with $x = 1$. If the program is not given permission to
execute ($ok = false$), then $x$ can taken any value.

Designs have a natural notion of refinement which requires that the precondition is weakened, and the postcondition
strengthened within the window of the precondition, as shown by the theorem below.

\begin{theorem} $\design{P_1}{P_2} \refinedby \design{Q_1}{Q_2} ~ \iff ~ (P_1 \implies Q_1) \land (Q_2 \land P_1 \implies P_2)$ \isalink{https://github.com/isabelle-utp/utp-main/blob/07cb0c256a90bc347289b5f5d202781b536fc640/theories/designs/utp_des_core.thy\#L584}
\end{theorem}

Design relations are closed under sequential composition, disjunction, and conjunction~\cite{Hoare&98,Cavalcanti04}, all
of which retain the denotations given in Definition~\ref{def:utp-prog} with the alphabet containing $ok$ and $ok'$. The
main design healthiness conditions are \healthy{H} and \healthy{N}, which are given below~\cite{Hoare&98,Cavalcanti04}.
\begin{definition}[Design Healthiness Conditions] $ $ \isalink{https://github.com/isabelle-utp/utp-main/blob/07cb0c256a90bc347289b5f5d202781b536fc640/theories/designs/utp_des_healths.thy}

  \begin{center}
  \begin{tabular}{ccc}
  {$\begin{aligned}
    \healthy{H1}(P) ~~\defs~~& (ok \implies P) \\
    \healthy{H2}(P) ~~\defs~~& P \relsemi \ckey{J} \\
    \healthy{H3}(P) ~~\defs~~& P \relsemi \IID \\
  \end{aligned}$} & ~~ &
  {$\begin{aligned}
    \ckey{J}        ~~\defs~~& (ok \implies ok') \land v' = v \\
    \IID            ~~\defs~~& \design{\true}{\II} \\
    \healthy{H}     ~~\defs~~& \healthy{H1} \circ \healthy{H2} \\
    \healthy{N}     ~~\defs~~& \healthy{H1} \circ \healthy{H3}
  \end{aligned}$}
  \end{tabular}
\end{center}
\end{definition}
$\healthy{H1}$ states that until a design has been given permission to execute, as recorded via $ok$, no observations
are possible. $\healthy{H2}$ states that no design can require non-termination. A more intuitive characterisation of the
$\healthy{H2}$ fixed-points is $P[false/ok'] \implies P[true/ok']$: every non-terminating behaviour of $P$ for which
$ok' = false$ has an equivalent terminating behaviour for which $ok' = true$. The composition, $\healthy{H}$, precisely
characterises the set of design relations constructed using $\design{P_1}{P_2}$~\cite{Cavalcanti04}.

$\healthy{H3}$ designs additionally require that $P$ is a condition: it does not refer to dashed variables. This
subclass of designs is useful for ``normal'' specifications, where the precondition does not refer to the final
state. $\ckey{H3}$ designs, with a few notable exceptions, are the most common form of design, and are thus sometimes
known as \emph{normal designs}~\cite{Guttman2010}, as indicated by healthiness condition $\healthy{N}$. Since every
$\healthy{H3}$ predicate is also $\healthy{H2}$ healthy, in defining $\healthy{N}$, we do not need to include
$\healthy{H2}$ in the composition.

$\healthy{H}$ and $\healthy{N}$ are both idempotent and continuous, and thus the theories they define are both complete
lattices. The bottom element $\botD$ is abortive, arising, for instance due to a violated precondition, and the top
$\topD$ is miraculous. The infimum $P \sqcap Q$ is a non-deterministic choice between two designs, and refinement
reduces non-determinism: $P \sqcap Q \refinedby P$.

\subsection{Reactive Processes}

The theory of reactive processes~\cite{Hoare&98,Cavalcanti&06} unifies the semantics of different reactive
languages. The two main goals of reactive processes are to (1) embed traces into the relational calculus, which is
achieved through healthiness conditions \healthy{R1} and \healthy{R2}, and (2) introduce intermediate observations,
which is achieved through healthiness condition \healthy{R3}. In addition to $ok$ and $ok'$, the theory has three pairs
of observational variables:
\begin{enumerate}
  \item $wait, wait' : \Bool$ that determine whether a process (or its predecessor) is waiting for interaction with its
    environment, that is, it is quiescent, or else has correctly terminated;
  \item $tr, tr' : \seq\,\textit{Event}$ that describes the trace before and after the process' execution; and 
  \item $\refu, \refu' : \power \textit{Event}$ that describe the events being refused during a quiescent state, as
    required by the failures-divergences model of CSP~\cite{Hoare85,Roscoe2005}. If all events are being refused, the
    process is in a deadlock situation.
  \end{enumerate}
Since reactive programs often run indefinitely, the theory of reactive processes distinguishes good and bad
non-termination, the being characterised by divergence. This is achieved by reinterpreting $ok$ to indicate
divergence. Specifically, if $ok'$ is false then a reactive process has diverged, meaning it is exhibiting
unpredictable or erroneous behaviour. If $ok'$ is true, and $wait'$ is false, the process has not terminated, but
neither has it diverged.

The theory of reactive processes is used to provide a UTP denotational semantics for both
CSP~\cite{Hoare87,Cavalcanti&06}, based on relational encoding of the failures-divergences model~\cite{Roscoe2005}, and
also the stateful process language \Circus~\cite{Woodcock2001-Circus,Oliveira2005-PHD,Oliveira&09}. \Circus provides all
the usual operators of CSP for expressing networks of communicating processes, together with state-based constructs such
as variable assignment. \Circus processes encapsulate a number of state variables, operations that act on those
variables, and actions that encode the reactive behaviour of the process using channels.

In previous work~\cite{Foster17a}, we have generalised the standard UTP theory of reactive
processes~\cite{Hoare&98}. Our generalised theory~\cite{Foster17a} removes the $\refu$ and $\refu'$ variables, which
allows us to characterise behavioural semantic models other than failures-divergences. Moreover, we add
$\state, \state' : \Sigma$ to explicitly model state as suggested by \cite{BGW09}, where $\Sigma$ is a state space
type. In our previous work~\cite{Foster17a}, we have shown how the UTP theory of reactive processes can be generalised
by characterising the trace model with an abstract algebra, called a ``trace algebra''.  We characterise traces with an
abstract set $\tset$ equipped with two operators: trace concatenation $\tcat : \tset \to \tset \to \tset$, and the empty
trace $\tempty : \tset$, which obey the following axioms~\cite{Foster17a}.
\begin{definition}\label{def:tralg} A trace algebra $(\tset, \tcat, \tempty)$ is a cancellative monoid satisfying the
  following axioms: \isalink{https://github.com/isabelle-utp/utp-main/blob/90ec1d65d63e91a69fbfeeafe69bd7d67f753a47/theories/reactive/Trace_Algebra.thy}
\begin{align*}
  x \tcat (y \tcat z) &= (x \tcat y) \tcat z \tag{TA1} \label{law:tassoc} \\
  \tempty \tcat x = x \tcat \tempty &= x \tag{TA2} \label{law:tunit} \\
  x \tcat y = x \tcat z ~~\implies~~  y &= z \tag{TA3} \label{law:tcancl1} \\
  x \tcat z = y \tcat z ~~\implies~~  x &= y \tag{TA4} \label{law:tcancl2} \\
  x \tcat y = \tempty  ~~\implies~~  x &= \tempty \tag{TA5} \label{law:tnai}
\end{align*}
\end{definition}
\noindent An example model is formed by finite sequences, $\langle a, b, \cdots, c \rangle$, that is
$(\seq A, \cat, \langle\rangle)$ forms a trace algebra, where $\cat$ is sequence concatenation. Using the two trace
algebra operators, we can also define a trace prefix operator ($x \le y$), and trace difference ($x - y$), which removes
a prefix $y$ from $x$. From these algebraic foundations, we have reconstructed the complete theory of reactive
processes, including its healthiness conditions and associated laws, in particular those for sequential and parallel
composition~\cite{Foster17a}. We thus generalise the type of $tr$ and $tr'$ to be an instance of a trace algebra
$\tset$, and recreate the three reactive healthiness conditions~\cite{Hoare&98,Cavalcanti&06}.
\begin{definition}[Stateful Reactive Healthiness Conditions] \label{def:reahealths} \isalink{https://github.com/isabelle-utp/utp-main/blob/07cb0c256a90bc347289b5f5d202781b536fc640/theories/reactive/utp_rea_healths.thy}
\begin{align*}
  \healthy{R1}(P)   & ~~\defs~~ P \land tr \le tr' \\
  \healthy{R2}_c(P)   & ~~\defs~~ \conditional{P[\tempty, \trace /tr, tr']}{tr \le tr'}{P} \\
  \healthy{R3}_h(P) & ~~\defs~~ \conditional{\IIsrd}{wait}{P} \\
  \IIsrd            & ~~\defs~~ \conditional{(\conditional{(\exists st @ \II)}{wait}{\II})}{ok}{\ckey{R1}(\true)} \\
  \trace            & ~~\defs~~ (tr' - tr) \\
  \healthy{R}_s     & ~~\defs~~  \healthy{R1} \circ \healthy{R2}_c \circ \healthy{R3}_h
\end{align*}
\end{definition}
$\healthy{R1}$ states that $tr$ is monotonically increasing; processes are not permitted to undo past
events. $\healthy{R2}_c$ is a version of \healthy{R2}~\cite{Hoare&98}, created to overcome an issue with definedness of
sequence difference~\cite{Foster17a}, but semantically equivalent in the context of \healthy{R1}. It states that a
process must be history independent: the only part of the trace it may constrain is $tr' \tminus tr$, that is, the
portion since the previous observation $tr$. Specifically, if the history is deleted, by substituting $\tempty$ for
$tr$, and $tr' - tr$ for $tr'$, then the behaviour of the process is unchanged. Our formulation of $\healthy{R2}_c$
deletes the history only when $tr \le tr'$, which ensures that $\healthy{R2}_c$ does not depend on $\healthy{R1}$, and
thus commutes with it. Intuitively, an $\healthy{R1}$-$\healthy{R2}_c$ healthy predicate syntactically does not
constrain the trace history ($tr$), but only the trace contribution expression ($\trace$), as the following theorem
illustrates.
\begin{theorem}[$\healthy{R1}$-$\healthy{R2}_c$ trace contribution] \label{thm:trcontr} \isalink{https://github.com/isabelle-utp/utp-main/blob/07cb0c256a90bc347289b5f5d202781b536fc640/theories/reactive/utp_rea_healths.thy\#L476}
$$\healthy{R1}(\healthy{R2}_c(P)) = (\exists t @ P[\tempty,t/tr,tr'] \land tr' = tr \tcat t)$$
\end{theorem}
Finally, we have $\healthy{R3}_h$, a version of \healthy{R3} from~\cite{BGW09} that introduces the concept of
intermediate observations, whilst ensuring that state variables are not included. $\healthy{R3}_h$ states that if a
process observes $wait$ to be true, then its predecessor has not yet terminated and thus it should behave like the
reactive identity, $\IIsrd$. For example, in a composition $P \relsemi Q$, if $P$ has not terminated then $Q$, if
$\healthy{R3}_h$ healthy, will behave as $\IIsrd$.

The reactive identity maintains the present value of all variables, other than the state $st$, when the predecessor is
in an intermediate state, or behaves like $\healthy{R1}(\true)$ if $ok$ is false. The latter scenario means that the
predecessor has diverged and thus we can guarantee nothing other than that the trace increases. Intuitively, an
$\healthy{R3}_h$ process conceals the state of any predecessor in an intermediate state. This allows that several
independent state valuations are concurrently possible, yet concealed from one another, until an observation is made
through an event interaction.

For comparison, we recall the definition of healthiness condition $\healthy{R3}$, which was previously used in the
theories for both CSP~\cite{Hoare&98,Cavalcanti&06} and \Circus~\cite{Oliveira&09}.
\begin{definition}[$\healthy{R3}$ Healthiness Condition] \isalink{https://github.com/isabelle-utp/utp-main/blob/07cb0c256a90bc347289b5f5d202781b536fc640/theories/reactive/utp_rea_healths.thy\#L658}
$$\healthy{R3}(P) ~\defs~ \conditional{\IIrea}{wait}{P} \qquad \qquad
  \IIrea ~\defs~ \conditional{\II}{ok}{\ckey{R1}(\true)}$$
\end{definition}
The only difference from $\healthy{R3}_h$ is that the identity $\IIrea$ is used in intermediate states. This operator
does not give special treatment to state variables: they are simply identified in intermediate states like other
observational variables. As discussed in detail in Section~\ref{sec:grd}, $\healthy{R3}_h$ allows a simpler treatment of
state variables, supports additional algebraic laws for assignment and state substitution, and solves the problem with
external choice for which it was originally designed~\cite{BGW09}, though at the cost of losing McEwan's interruption
operator~\cite{McEwan06}. Thus, the use of $\healthy{R3}_h$ instead of $\healthy{R3}$ is a design decision based on the
particular modelling facilities of interest.

We compose the three constituents to yield $\healthy{R}_s$, the overall healthiness condition of (stateful) reactive
processes, which is idempotent and continuous.
\begin{theorem}[Reactive Process Theory Properties] \label{thm:reaprop} \isalink{https://github.com/isabelle-utp/utp-main/blob/07cb0c256a90bc347289b5f5d202781b536fc640/theories/rea_designs/utp_rdes_healths.thy\#L226}
  \begin{itemize}
    \item $\healthy{R}_s$ is idempotent: $\Rs(\Rs(P)) = \Rs(P)$;
    \item $\healthy{R}_s$ is continuous: $\Rs(\bigsqcap A) = \bigsqcap P \in A @ \Rs(P)$.
    \end{itemize}
  \end{theorem}
\noindent As for designs, a corollary of this theorem is that we obtain a complete lattice and the continuous theory
properties of Theorem~\ref{thm:contthm}. Thus we now have a UTP theory of stateful reactive processes to use as the
foundation for reactive design contracts.

%% file: sec/rrel.tex
In this section, we begin the main novel contributions of our paper, by introducing a theory of reactive relations that
we use to describe assumptions and guarantees in our reactive contracts in Section~\ref{sec:contracts}. A reactive
relation is an $\healthy{R1}$-$\healthy{R2}_c$-healthy predicate that does not have $ok$, $ok'$, $wait$, and $wait'$ in
its alphabet. Such a relation is effectively an alphabetised relation with the non-relational trace variable $\trace$
present. We define the following healthiness condition for reactive relations.
\begin{definition}[Reactive Relations] \label{def:RR} \isalink{https://github.com/isabelle-utp/utp-main/blob/07cb0c256a90bc347289b5f5d202781b536fc640/theories/reactive/utp_rea_rel.thy\#L16}
  \begin{align*}
    \healthy{RR}(P) = (\exists ok, ok', wait, wait' @ \healthy{R1}(\healthy{R2}_c(P)))
  \end{align*}
\end{definition}
\noindent $\healthy{RR}$ restricts access to $ok$ and $wait$ through existential quantification. In general, if
$(\exists x @ P) = P$ then it must be the case that $P$ does not refer to $x$. With the help of
Theorem~\ref{thm:reaprop}, we can show that $\healthy{RR}$ is both idempotent and continuous.
\begin{theorem} $\healthy{RR}$ is idempotent and continuous. \isalink{https://github.com/isabelle-utp/utp-main/blob/07cb0c256a90bc347289b5f5d202781b536fc640/theories/reactive/utp_rea_rel.thy\#L22} 
\end{theorem}

\noindent We can therefore also show that reactive relations form a complete lattice.

\begin{theorem} $\theoryset{\healthy{RR}}$ forms a complete lattice, with bottom element $\healthy{R1}(\true)$, and top
  element $\false$. \isalink{https://github.com/isabelle-utp/utp-main/blob/07cb0c256a90bc347289b5f5d202781b536fc640/theories/reactive/utp_rea_rel.thy\#L741}
\end{theorem}

\begin{proof}
  We obtain a complete lattice by Knaster-Tarski~\cite{Tarski71}, and by Theorem~\ref{thm:contthm} the bottom and top
  elements are $\healthy{RR}(\true)$ and $\healthy{RR}(\false)$, respectively. For illustration, we give the following
  calculation that shows how the former reduces to $\healthy{R1}(\true)$.
\begin{align*}
  \healthy{RR}(\true) 
  &= (\exists ok, ok', wait, wait' @ \healthy{R1}(\healthy{R2}_c(\true))) & [\ref{def:RR}] \\
  &= (\exists ok, ok', wait, wait' @ \healthy{R1}(\conditional{\true[\tempty, \trace /tr, tr']}{tr \le tr'}{\true})) & [\ref{def:reahealths}] \\
  &= (\exists ok, ok', wait, wait' @ \healthy{R1}(\conditional{\true}{tr \le tr'}{\true})) & [\text{vacuous substitution}] \\
  &= (\exists ok, ok', wait, wait' @ \healthy{R1}(\true)) & [\text{relational calculus}] \\
  &= (\exists ok, ok', wait, wait' @ tr \le tr') & [\ref{def:reahealths}] \\
  &= (tr \le tr') & [\text{predicate calculus}] \\
  &= \healthy{R1}(\true) & [\ref{def:reahealths}]
\end{align*}
Calculation of $\healthy{RR}(\false) = \false$ follows a similar form.
\end{proof}

\noindent Here, $\healthy{R1}(\true)$ is the most non-deterministic relation where the trace is monotonically
increasing. As for relations, $\false$ is miraculous reactive relation with no possible observations.

Since reactive relations are a kind of condition, it is useful to have an associated Boolean algebra to support contract
and specification construction. However, logical negation is not closed under $\healthy{R1}$ and thus it is necessary to
redefine negation, and also implication, for similar reasons, for reactive relations.

\begin{definition}[Reactive Relation Logical Operators] \isalink{https://github.com/isabelle-utp/utp-main/blob/07cb0c256a90bc347289b5f5d202781b536fc640/theories/reactive/utp_rea_rel.thy\#L95}
  $$\true_r ~~\defs~~ \healthy{R1}(\true) \qquad \negr P ~~\defs~~ \healthy{R1}(\neg P) \qquad P \implies_r Q ~~\defs~~ (\negr P \lor Q)$$
\end{definition}

\noindent The universal relation $\true$ is not $\healthy{RR}$-healthy, since it allows any combination of $tr$ and
$tr'$. Consequently, we define $\truer$, which is the bottom element. Reactive negation, $\negr P$, negates $P$ and then
applies \healthy{R1}. Effectively this yields a predicate whose corresponding set of trace extensions does not satisfy
$P$. Since $\healthy{RR}$ is already closed under the other Boolean operators, such as $\lor$, $\land$, and $\false$, we
can apply them directly and prove the following theorem.
\begin{theorem} $(\healthy{RR}, \land, \lor, \negr, \truer, \false)$ forms a Boolean algebra. \isalink{https://github.com/isabelle-utp/utp-main/blob/07cb0c256a90bc347289b5f5d202781b536fc640/theories/reactive/utp_rea_rel.thy\#L574}
\end{theorem}
We can also prove the following closure properties for the standard relational operators of Definition~\ref{def:utp-prog}:
\begin{theorem}[Relational Operators Closure Properties] \isalink{https://github.com/isabelle-utp/utp-main/blob/90ec1d65d63e91a69fbfeeafe69bd7d67f753a47/theories/reactive/utp_rea_rel.thy\#L341}
  \begin{itemize}
    \item If $P$ and $Q$ are $\healthy{RR}$ then $P \relsemi Q$ is $\healthy{RR}$;
    \item If $I \neq \emptyset$ and $\forall i @ P(i)$ is $\healthy{RR}$, then $\bigsqcap_{i \in I} \, P(i)$ is $\healthy{RR}$.
  \end{itemize}
\end{theorem}

\noindent $\healthy{RR}$ is closed under sequential composition ($\relsemi$) and nondeterministic choice
($\bigsqcap$). Consequently, we can reuse many of the corresponding algebraic laws of the alphabetised relational
calculus listed in Section~\ref{sec:utp}. This is a significant advantage to the UTP approach of conservatively
extending existing theories. However, the relational assignment operator is not healthy, because we can use it to
perform arbitrary updates to the trace; $tr := \tempty$ is not $\healthy{R1}$ healthy, for instance. Consequently, we need
to define new operators for manipulating and querying a reactive program's state ($\Sigma$) via observational variable
$st : \Sigma$. We therefore also define the following operators:
\begin{definition}[Reactive Relational State Operators Assignment] \isalink{https://github.com/isabelle-utp/utp-main/blob/07cb0c256a90bc347289b5f5d202781b536fc640/theories/reactive/utp_rea_prog.thy\#L232}
  \begin{align*}
    \assignsr{\sigma} &\defs (tr' = tr \land st' = \sigma(st) \land v' = v) \\
    \IIr              &\defs \assignsr{id} \\
    \rstpred{s}       &\defs \healthy{R1}(s) & \text{provided $s$ refers to undashed state variables only}
  \end{align*}
\end{definition}
\noindent $\assignsr{\sigma}$ is an assignment operator, in the style of Back's update action~\cite{Back1998}, that
applies a substitution function $\sigma : \Sigma \to \Sigma$ to the state-space variable $st$, and leaves all other
variables unchanged. Since the alphabet is open, we use the shorthand $v$ to refer to the variable set excluding $ok$,
$wait$, $tr$, and $st$. Substitutions functions can be constructed using the notation
$\{x_1 \mapsto v_1, \cdots, x_n \mapsto v_n\}$ which associates $n$ expressions ($v_i$) to corresponding variables
($x_i$). A singleton assignment can be denoted as $\assignr{x}{v} \defs \assignsr{\{x \mapsto v\}}$, where $v$ is an
expression on undashed state variables.

As usual, we also introduce the degenerate form $\IIr$, which simply retains the values of all variables. We also define
a state condition operator $\rstpred{s}$, where $s$ is a predicate over undashed state variables only: it is a condition
not mentioning variables $ok$, $wait$, or $tr$. The operator requires that $s$ holds on the state variables, whilst
leaving the trace unconstrained. We can demonstrate the following healthiness properties for these operators.

\begin{theorem} $\assignsr{\sigma}$, $\IIr$, and $\rstpred{s}$ are $\healthy{RR}$-healthy \isalink{https://github.com/isabelle-utp/utp-main/blob/07cb0c256a90bc347289b5f5d202781b536fc640/theories/reactive/utp_rea_prog.thy\#L245}
\end{theorem}

Assignment is $\healthy{RR}$, since we conjoin with $tr' = tr$, which is $\healthy{R1}$ and $\healthy{R2}_c$ healthy,
and do not refer to $ok$ or $wait$. $\IIr$ is $\healthy{RR}$ for the same reasons. The state condition is \healthy{RR}
healthy as it is clearly $\healthy{R1}$, and also it is $\healthy{R2}_c$ since $s$ contains no reference to trace
variables.

A useful subset of the reactive relations is the reactive conditions, which we use to encode contractual
preconditions. A relational condition $b$ is a relation that does not refer to dashed variables. Such conditions can be
characterised as fixed points of the idempotent function $\healthy{C}(b) \defs b \relsemi \true$. For example, the
precondition of a $\healthy{H3}$ design is a condition on the initial state variables only, and so is
$\healthy{C}$-healthy. For reactive relations, we cannot exclude all dashed variables as we wish to express trace
constraints using $\trace$, which includes $tr$ and $tr'$. Consequently, reactive conditions are characterised by the
following healthiness condition, $\healthy{RC}$.

\begin{definition}[Reactive Conditions] \label{def:RC} \isalink{https://github.com/isabelle-utp/utp-main/blob/07cb0c256a90bc347289b5f5d202781b536fc640/theories/reactive/utp_rea_cond.thy\#L9}
  $$\healthy{RC1}(P) ~~\defs~~ \negr ((\negr P) \relsemi \true_r) \qquad \quad
    \healthy{RC} ~~\defs~~ \healthy{RC1} \circ \healthy{RR}$$
\end{definition}

\noindent We require that $\true_r$ is a right unit of the predicate's negated form, which means, firstly, that it can
refer only to undashed state and observational variables other than $tr$. Secondly, the behaviour of $tr$ is restricted
by having $tr \le tr'$ as a right unit. Intuitively, this means that a reactive condition's complement is extension
closed~\cite{Roscoe2005}: if trace $t_1$ is permitted by $\negr P$ then for any trace $t_2$, $t_1 \cat t_2$ is also
permitted. Extension closure is here characterised by effectively requiring that $\truer$ is a right unit of $\negr P$.

The reason for this constraint is that if a trace violates a reactive condition, that is the precondition of a reactive
contract, then any extension should also violate it. A reactive condition is technically a relation, but can be
considered as a condition on the state variables and the trace variable ($\trace$). We can show, for example, that any
state condition $\rstpred{s}$, which does not constrain $\trace$, is $\healthy{RC1}$ by the following calculation:
\begin{example}[State Condition is \healthy{RC1}] \isalink{https://github.com/isabelle-utp/utp-main/blob/07cb0c256a90bc347289b5f5d202781b536fc640/theories/reactive/utp_rea_prog.thy\#L643}
\begin{align*}
  \healthy{RC1}(\rstpred{s})
  ~~=~~& \negr ((\negr \rstpred{s}) \relsemi \truer) & \text{[\ref{def:RC}]} \\
  ~~=~~& \negr (\rstpred{\lnot s} \relsemi \truer) & \text{[predicate calculus]} \\
  ~~=~~& \negr ((tr \le tr' \land \lnot s) \relsemi tr \le tr') & \text{[$\healthy{R1}$, $\truer$, $\rstpred{-}$ definitions]}\\
  ~~=~~& \negr (tr \le tr' \relsemi tr \le tr' \land (\lnot s) \relsemi tr \le tr') & \text{[distribution: $\neg s$ is condition]} \\
  ~~=~~& \negr (tr \le tr' \land (\lnot s) \relsemi tr \le tr') & \text{[transitivity of $\le$]}\\
  ~~=~~& \negr (tr \le tr' \land (\exists t_0 @ (\lnot s) \land t_0 \le tr')) & \text{[$\relsemi$ definition, substitution]} \\ 
  ~~=~~& \negr (tr \le tr' \land \lnot s) & \text{[predicate calculus]}\\
  ~~=~~& \negr (\negr \rstpred{s}) & \text{[$\negr$, $\healthy{R1}$ definitions]} \\
  ~~=~~& \rstpred{s} & \text{[double negation]}
\end{align*}

\noindent The calculation first pushes the negation into the state condition, to yield $\rstpred{\neg s}$. Since this is
$\healthy{R1}$, but does not otherwise constrain $tr$ and $tr'$, any extension of the trace is permitted. Consequently,
$\truer$ is a right unit of $\rstpred{\neg s}$, and then, by relational calculus, $\rstpred{s}$ is $\healthy{RC1}$
healthy. \qed
\end{example}
Reactive conditions can also constrain $tr'$, but only if the corresponding trace extension $\trace$ refers only to a
prefix of the trace, leaving the suffix unconstrained. Consider, for example, $\negr (\langle a \rangle \le \trace)$, a
reactive relation that forbids $a$ from being the first element of $\trace$. It permits the empty trace $\tempty$, and any
trace $\langle b, \cdots \rangle$, where $b \neq a$. It forbids the trace $\langle a \rangle$ and any extension
thereof. It is \healthy{RC} healthy, because its negated form is extension closed, as confirmed below.
\begin{example}[Constrained Prefix is \healthy{RC1}] \label{ex:constrprf} \isalink{https://github.com/isabelle-utp/utp-main/blob/07cb0c256a90bc347289b5f5d202781b536fc640/theories/reactive/utp_rea_cond.thy\#L212}
\begin{align*}
  \healthy{RC1}(\negr (\langle a \rangle \le \trace))
  ~~=~~& \negr ((\negr \negr (\langle a \rangle \le \trace)) \relsemi \true_r) & \text{[\healthy{RC1} definition]} \\
  ~~=~~& \negr ((\langle a \rangle \le \trace) \relsemi \true_r) & \text{[double negation]} \\
  ~~=~~& \negr (tr \cat \langle a \rangle \le tr' \relsemi tr \le tr') & \text{[$\trace$, $\truer$ definition]} \\
  ~~=~~& \negr (tr \cat \langle a \rangle \le tr') & \text{[composition of $\le$]} \\
  ~~=~~& \negr (\langle a \rangle \le \trace) & \text{[$\trace$ definition]}
\end{align*}
\noindent Crucially, the relation $tr \cat \langle a \rangle \le tr'$ is extension closed, and consequently has $\truer$ as a right
unit. \qed
\end{example}
Reactive conditions thus serve to restrict permissible initial behaviours in the trace; the previous example states that
the event $a$ must not be performed initially. Thus, an alternative characterisation of reactive conditions is that the
trace is prefix closed, which can be characterised by the following healthiness condition.

\begin{definition} $\healthy{RC2}(P) \defs \healthy{R1}(P \relsemi tr' \le tr)$ \isalink{https://github.com/isabelle-utp/utp-main/blob/07cb0c256a90bc347289b5f5d202781b536fc640/theories/reactive/utp_rea_cond.thy\#L72}
\end{definition}

\noindent $\healthy{RC2}$ first sequentially composes $P$ with $tr' \le tr$, which is the converse of $\truer$, and
states that the trace monotonically decreases. This has the effect of abstracting references to variables other than
$tr$, and recording every trace which is a prefix of the traces $tr'$ produced by $P$. Then, $\healthy{R1}$ is applied to
remove traces that are shorter than those of the initial $tr$ passed to $P$. The intuition is given by the following theorem:

\begin{theorem} If $P \is \healthy{RR}$ then $\healthy{RC2}(P) = (\exists (t_0, t_1) @ (\exists \state' @ P[\tempty,t_1/tr,tr'] \land t_0 \le t_1 \land tr' = tr \cat t_0))$ \isalink{https://github.com/isabelle-utp/utp-main/blob/90ec1d65d63e91a69fbfeeafe69bd7d67f753a47/theories/reactive/utp_rea_cond.thy\#L94}
\end{theorem}

Here, $t_1$ is one of the traces contributes by $P$, and $t_0$ is arbitrary prefix of $t_1$. Application of
$\healthy{RC2}$ inserts all such $t_1$ traces for every $t_0$ trace, and constructs the overall trace to be the $tr$
extended by $t_0$. If adding these prefixes as observations has no effect, because they are present already, then the
reactive relation is $\healthy{RC2}$ healthy. Though $\healthy{RC2}$ is not identical to $\healthy{RC1}$, we can show
that it has the same set of fixed points.

\begin{theorem} If $P \is \healthy{RR}$, then $P \is \healthy{RC1}$ if and only if $P \is \healthy{RC2}$ \isalink{https://github.com/isabelle-utp/utp-main/blob/07cb0c256a90bc347289b5f5d202781b536fc640/theories/reactive/utp_rea_cond.thy\#L82}
\end{theorem}

The theorem shows that a reactive relation is prefix closed when its complement is extension closed, and vice-versa. The
intuition is that if reactive condition $P$ admits $t_1$, then it must also admit any prefix $t_0 \le t_1$. If $t_0$ was
excluded from $P$ then any extension, including $t_1$, would also be excluded, since $\negr P$ is extension closed,
contradicting our assumption. We can therefore use $\healthy{RC2}$ healthiness to demonstrate $\healthy{RC1}$
healthiness. Both healthiness conditions are idempotent and continuous, and consequently $\healthy{RC}$ predicates form
a complete lattice. In particular, we retain the lattice top and bottom elements $\false$ and $\true_r$, and also the
connectives $\land$ and $\lor$. However, $\healthy{RC}$ predicates are not closed under reactive negation, since this
does not preserve prefix closure. The unrestricted use of negation in this context, however, is not necessary for the
purposes of this paper.

We also define a reactive weakest liberal precondition operator~\cite{Dijkstra75,Hoare78}.

\begin{definition}[Reactive Weakest (Liberal) Precondition] \label{def:wpR} \isalink{https://github.com/isabelle-utp/utp-main/blob/90ec1d65d63e91a69fbfeeafe69bd7d67f753a47/theories/reactive/utp_rea_wp.thy\#L10}
  \begin{align*}
    P \wpR Q ~~\defs~~&  \neg_r\, (P \relsemi \neg_r\, Q) & \text{provided $P$ is $\healthy{RR}$ and $Q$ is $\healthy{RC}$}
  \end{align*}
\end{definition}

\noindent 

\noindent The predicate has the usual intuition: $P \wpR Q$ is the weakest reactive condition such that if reactive
relation $P$ terminates, it achieves a final observation satisfying reactive condition $Q$. The definition is similar to
that given for relations in~\cite{Hoare&98,Cavalcanti04}, which effectively takes the complement of the observations
under which $P$ fails to establish $Q$. We have simply replaced relational negation with reactive negation.  

The predicate $P \wpR Q$ is a reactive condition ($\healthy{RC}$) provided that $P$ is reactive relation and $Q$ is a
reactive condition. Although, we are using complement, which does not retain prefix closure, we apply it twice which
leads to restoration of prefix closure in the final form.

From this definition, we can prove a number of standard \ckey{wlp} laws~\cite{Dijkstra75,Hoare78}, which we enumerate
below.

\begin{theorem}[Reactive Weakest Precondition Laws] \isalink{https://github.com/isabelle-utp/utp-main/blob/90ec1d65d63e91a69fbfeeafe69bd7d67f753a47/theories/reactive/utp_rea_wp.thy\#L41}
  \begin{align}
    P \wpR \true_r ~~=~~& \true_r \\
    P \wpR (Q \land R) ~~=~~& (P \wpR Q \land P \wpR R) \\
    (P ~\infixIf b \infixElseR~ Q) \wpR R ~~=~~& (P \wpR R ~ \infixIf b \infixElseR ~ Q \wpR R) \\
    (P \relsemi Q) \wpR R  ~~=~~ & P \wpR (Q \wpR R) \\
    \IIr \wpR R ~~=~~ & R \\
    \assignsr{\sigma} \wpR R ~~=~~ & \substapp{\sigma}{R} \label{thm:rwp-assigns} \\
    \false \wpR P ~~=~~ & \true_r \\
    (P \intchoice Q) \wpR R ~~=~~ & (P \wpR R) \land (Q \wpR R) \\
    \left(\bigsqcap i \in A @ P(i)\right) \wpR R ~~=~~ & \left(\forall i \in A @ P(i) \wpR R\right)
  \end{align}
\end{theorem}

\noindent These laws are similar to those given by Dijkstra~\cite{Dijkstra75} and Hoare~\cite{Hoare78}. In particular,
we note that the miraculous reactive relation $\false$ has the weakest liberal precondition $\truer$, which is why it is
``liberal'': we can make no judgements about a non-terminating reactive relation. The assignment law
(Theorem~\ref{thm:rwp-assigns}) uses a substitution operator $\substapp{\sigma}{R}$ to apply substitution function
$\sigma$ to predicate $R$. In words, $\assignsr{\sigma}$ achieves $R$ provided that $R$ holds when all its variables are
replaced by those given in the assignment $\sigma$.

We have now constructed a model for simple reactive programs containing both traces, assignments, and conditions. Like
UTP relations, reactive relations do not have the expressivity to account for non-terminating reactive behaviours. These
are accounted for by our theory of reactive contracts, which we define the next section.

%% file: sec/contracts.tex
In this section, we describe the signature of our theory of reactive design contracts and algebraic theorems. Due to its
complexity, we defer the definition of the UTP theory's healthiness condition (\healthy{NSRD}) until
Section~\ref{sec:grd}. As we have mentioned in the introduction, reactive programs can be denoted as contracts that
represent their assumptions and guarantees. Our goal is to provide a general method for calculating the contract of
reactive program, supported by equational theorems that can reduce a composition of multiple contracts into a single
unified contract specification, which can then be subjected to verification. All the laws we present are mechanically
proven theorems of our UTP theory; here we also provide some intuition for why they hold. We illustrate the use of our
contract notation with a number of \Circus-based~\cite{Woodcock2001-Circus} examples, which give intuition, though
stateful failure-divergences is not the only applicable semantic model.

\subsection{Contracts and Refinement}
\label{sec:contref}

Reactive program components normally proceed through three phases during execution:

\begin{enumerate}
  \item \textbf{pre-execution} -- the program waits for its predecessor to terminate and does not contribute any
    observable behaviour.
  \item \textbf{intermediate execution} -- the program begins the main body of its execution, which includes
    communication with other concurrent processes, and updates to its state. During this time state updates are, however,
    hidden from its successor.
  \item \textbf{termination} -- the program ceases interaction with the environment, reveals its final state to the
    successor, and signals permission for it to begin. Since reactive programs often do not terminate, this phase may
    never be reached.
\end{enumerate}

\noindent In this view, we largely assume that parallel programs do not directly share state but, as is typical in
process algebras, they must explicitly communicate using a suitable mechanism such as channels. All other activity, such
as state updates, is internalised to the sequential behaviour of the process, though it is possible to merge the state
of several terminated parallel processes~\cite{Oliveira&09}. Shared variables can, nevertheless, be modelled by encoding
them within traces.

Reactive programs can also diverge~\cite{Roscoe2005}, meaning they exhibit erroneous behaviour, such as engaging in an
infinite sequence of internal activity without any communication. Divergence corresponds to violation of a contract's
assumptions. A reactive design contract is a triple~\cite{Canham15} of the form
$$\rc{P(\state, \trace, r)}{Q(\state, \trace, r, r')}{R(\state, \state', \trace, r, r')}$$ the three parts of which are:

\begin{enumerate}
\item the \textbf{precondition} $P$, with assumptions the contract makes before it executes, violation of which
  corresponds to a programmer error such as divergence. It is a reactive condition, and can therefore refer to the
  initial state $\state$, the trace contribution $\trace$, and potentially other (unprimed) observational variables in the
  alphabet ($r$), but not observational variables $ok$ or $wait$, or primed variables other then $tr'$. Access
  to $tr'$ is usually indirect through $\trace$.
\item the \textbf{pericondition} $Q$, with commitments the contract guarantees to fulfil during its intermediate
  execution steps. Often, it used to represent ``quiescent'' observations, where the program is awaiting interaction with
  it environment. It is a reactive relation only on the initial states, $\trace$, and any other variables ($r, r'$).
  \item the \textbf{postcondition} $R$, with commitments that are fulfilled should the program terminate. It is a
    reactive relation that can additionally refer to the final state $\state'$, unlike the pre and pericondition.
\end{enumerate}

\noindent Such contracts can be used both as specifications, for encoding assumptions and guarantees for a subsystem, or
alternatively as a means to encode the semantics of a reactive programming language.  A reactive design contract has the
following definition.
\begin{definition}[Reactive Design Contract] \label{def:rcontract} \isalink{https://github.com/isabelle-utp/utp-main/blob/90ec1d65d63e91a69fbfeeafe69bd7d67f753a47/theories/rea_designs/utp_rdes_triples.thy\#L492}
 $$\rc{P_1}{P_2}{P_3} ~~\defs~~ \rtdes{P_1}{P_2}{P_3} \quad \text{where $P_1 \is \healthy{RC}$, and $P_2$ and $P_3$ are both $\healthy{RR}$}$$
\end{definition}
\noindent This definition assumes that $P_1$, $P_2$, and $P_3$ are as specified above. This is formalised by requiring
that $P_1$ is a reactive condition, and $P_2$ and $P_3$ are both reactive relations, using the theory developed in
Section~\ref{sec:rrel}. The reactive contract is a form of UTP design which is made reactive using $\Rs$. In previous
work~\cite{Oliveira&09}, reactive designs are often written in just two parts ($\Rs(P \shows Q)$), the assumption and
guarantee, with the intermediate and final behaviours intertwined. Here, we adopt the triple notation first developed in
\cite{Canham15} as it allows us to consider these separately and simplifies many laws. The diamond
$P_2 \wcond P_3$~\cite{Canham15} is simply an abbreviation for $\conditional{P_2}{wait'}{P_3}$, which distinguishes
intermediate and final non-divergent observations.

Our theory supports contract refinement, which is characterised by the following theorem:
\begin{theorem}[Reactive Design Refinement] \label{thm:rdesrefine} \isalink{https://github.com/isabelle-utp/utp-main/blob/90ec1d65d63e91a69fbfeeafe69bd7d67f753a47/theories/rea_designs/utp_rdes_triples.thy\#L846}
  $$\rc{P_1}{P_2}{P_3} \refinedby \rc{Q_1}{Q_2}{Q_3} \iff [P_1 \implies Q_1] \land (P_2 \refinedby (Q_2 \land P_1)) \land (P_3 \refinedby (Q_3
  \land P_1))$$
\end{theorem}
\noindent This is not a definition, but a theorem of the UTP refinement operator from Definition~\ref{def:refinement}
that is supported by the UTP theory (elaborated in Section~\ref{sec:grd}). Theorem~\ref{thm:rdesrefine} shows that
contract refinement reduces to three proof obligations:
\begin{enumerate}
  \item the precondition is weakened ($P_1 \implies Q_1$);
  \item the pericondition of the first contract ($P_2$) is strengthened by the pericondition of the second ($Q_2$),
    conjoined with the precondition of the first ($P_1$);
  \item the postcondition of the first contract ($P_3$) is strengthened by the postcondition of the second ($Q_3$),
    conjoined with the precondition of the first ($P_1$).
\end{enumerate}
Such a weakening of the assumption and strengthening of the guarantees, of course, is a defining feature of most
contract theories~\cite{Meyer92,Back1998,Benveniste2007,Benvenuti2008}. The particular value of this theorem in our
case is to provide a foundation for a standard verification procedure for contract-based reactive languages. If a
language can be given a contractual denotational semantics, meaning that every operator can be assigned a reactive
contract, then we can solve a verification problem, $\rc{\!P_1\!}{\!P_2\!}{\!P_3\!} \refinedby Q$, for a specification
$P$ and reactive program $Q$. 

We first calculate the program's contact $Q = \rc{\!Q_1\!}{\!Q_2\!}{\!Q_3\!}$, and then use Theorem~\ref{thm:rdesrefine}
to produce the three proof obligation predicates. Then, we can utilise theorem proving technology for relational
calculus in Isabelle/UTP~\cite{Foster14,Foster16a} to attempt discharge of the proof obligations. In Isabelle/HOL, this
can be supported by the \textsf{sledgehammer} proof method~\cite{Blanchette2011} that harnesses external automated
theorem provers. Consequently, our theorem of contracts can be used to support an automated verification technique for
reactive programs. This allows us, in particular, to support verification of programs and models with a very large or
infinite state space, since the calculated contracts are symbolic rather than explicit entities, which allows us to
overcome the state explosion problem. This, then, is the utility of the complex theory that follows.

In addition to the refinement law, we also have a similar theorem for proving equivalences:
\begin{theorem}[Reactive Design Equivalence] \label{thm:rdesequiv} \isalink{https://github.com/isabelle-utp/utp-main/blob/90ec1d65d63e91a69fbfeeafe69bd7d67f753a47/theories/rea_designs/utp_rdes_triples.thy\#L909}
  $$\rc{P_1}{P_2}{P_3} = \rc{Q_1}{Q_2}{Q_3} \iff (P_1 = Q_1) \land ((P_2 \land Q_1) = (Q_2 \land P_1)) \land ((P_3 \land Q_1) = (Q_3 \land P_1))$$
\end{theorem}

\noindent This theorem is a consequence of Theorem~\ref{thm:rdesequiv} and the fact that refinement is
antisymmetric. Two reactive contracts are equivalent if, and only if, (1) their preconditions are equivalent, and (2)
their peri- and postconditions are equivalent modulo the precondition. With this theorem we can similarly automate
proving equivalences.

\subsection{Denoting Reactive Programs}

A crucial requirement of the verification strategy outlined above is that the target language is equipped with a
denotational semantics in terms of reactive contracts. We illustrate the use of contracts in giving a denotational
semantics to a reactive language by denoting several of the operators from the \Circus
language~\cite{Woodcock2001-Circus,Oliveira&09}. For this, we first need to specialise the semantic model to
failure-divergences~\cite{Roscoe2005}. We thus specialise the trace algebra to finite sequences,
$tr : \seq\,\textit{Event}$, for some suitable set of events, and add the observational variable
$\refu' : \power\,\textit{Event}$, as usual~\cite{Cavalcanti&06}. This allows us to record the set of events which are
refused when an action is in a quiescent state, or equivalently the set of events that the action is willing to engage
in. It is equivalent to encoding CSP failure traces~\cite{Roscoe2005}, which consist of a sequence of events and a
refusal set. Adding observational variables is possible because the alphabet of the reactive design theory is
extensible; consequently our verification strategy is effectively parametric in a specialised semantic model.

We begin with the example of \Circus event prefix~\cite{Oliveira&09}, which is denoted by a reactive design triple:
\begin{example}[Event Prefix Reactive Design] \label{ex:prefix} \isalink{https://github.com/isabelle-utp/utp-main/blob/90ec1d65d63e91a69fbfeeafe69bd7d67f753a47/theories/sf_rdes/utp_sfrd_prog.thy\#L591}
  $$a \then \Skip ~~\defs~~ \rc{\truer}{a \notin \refu' \land \trace = \langle\rangle}{\state' = \state \land \trace = \langle a \rangle}$$

  \noindent This simple event prefix represents a program that, when enabled, waits for the environment to permit an $a$
  event, and following this, terminates. It is denoted by a contract with a true precondition since it can never
  diverge; every environment is a valid context. Its pericondition encodes a single quiescent observation: the event $a$
  is not refused and no events has been contributed to the trace as yet. The postcondition states that, when the program
  terminates, the state is unchanged by the event, and the trace is extended with $a$. In \Circus one can use this
  definition to represent the more general prefix construct using sequential composition:
  $a \then P ~~\defs~~ (a \then \Skip) \relsemi P$. \qed
\end{example}
Other examples are the $\Skip$ action, which represents a terminating process, and the $\Stop$ action, which represents a
deadlock.
\begin{example}[Terminated and Deadlocked Actions] \label{ex:termdead} \isalink{https://github.com/isabelle-utp/utp-main/blob/90ec1d65d63e91a69fbfeeafe69bd7d67f753a47/theories/sf_rdes/utp_sfrd_healths.thy\#L69}
  \begin{align*}
    \Skip ~~\defs~~& \rc{\truer}{\false}{\trace = \langle\rangle \land \state' = \state} \\
    \Stop ~~\defs~~& \rc{\truer}{\trace = \langle\rangle \land \refu' \subseteq \textit{Event}}{\false}
  \end{align*}

  \noindent The terminated action $\Skip$ has a true precondition. It has no intermediate observations, so the
  pericondition is $\false$, as it is essentially instantaneous and never pauses for interaction. In the postcondition,
  it is specified that the action makes no contribution to the trace, and leaves the state variables unchanged. The
  deadlocked action ($\Stop$) likewise has a true precondition. No state is a final state, indicated by the $\false$
  postcondition, since the process does not terminate. In the quiescent states it is simply required that the trace is
  unchanged, and any refusal set is observable, since no event is enabled.
  \qed
\end{example}

Our final example is external choice over a contract indexed by set $A$. 

\begin{example}[External Choice] \label{ex:extchoice} \isalink{https://github.com/isabelle-utp/utp-main/blob/90ec1d65d63e91a69fbfeeafe69bd7d67f753a47/theories/sf_rdes/utp_sfrd_extchoice.thy\#L124}
  \begin{align*}
    &\Extchoice i \in A @ \rc{P_1(i)}{P_2(i)}{P_3(i)} = \\[.5ex]
    &\qquad \rc{\bigwedge_{i \in A} P_1(i)}{\conditional{\left(\bigwedge_{i \in A} P_2(i)\right)}{\trace = \langle\rangle}{\left(\bigvee_{i \in A} P_2(i)\right)}}{\bigvee_{i \in A} P_3(i)}
  \end{align*}

  \noindent This is the first example of a contract composition law: it shows how a collection of contracts, in this
  case an indexed set, can be composed in a single contract. Such composition laws can then be combine with definitions,
  like those in Examples~\ref{ex:prefix} and \ref{ex:termdead} for contract calculation. The overall contract permits
  internal activity in the choice branches, but the choice itself is not resolved until an external event occurs. The
  precondition requires that the preconditions of all branches of the external choice hold in the initial state. In the
  pericondition, while the trace has not changed and thus no event has occurred ($\trace = \langle\rangle$), all
  periconditions of the choice hold simultaneously. Once an event has occurred only one of the periconditions need
  hold. This is the reason why the pericondition does not refer to final states, as these are concealed until
  termination or observation. Finally, in the postcondition, one of the choice branch postcondition holds. \qed
\end{example}
Though the three example definitions look different from the standard presentation of \Circus~\cite{Oliveira&09}, they
are largely equivalent. Indeed, the definitions given above are largely theorems of the original \Circus definitions,
and therefore our encoding is conservative. The exception is event prefix, in which we conceal the state whilst waiting
for the event, following previous work~\cite{BGW09}.

We now have a contractual denotational semantics for simple \Circus actions. For verification, we also need to specify
properties for reactive programs using specification contracts. A common desirable property of \Circus actions and CSP
processes is deadlock-freedom~\cite{Roscoe2005}, which states that a process never reaches a quiescent state where no
event is enabled. It can be specified using the following reactive contract:

\begin{definition}[Deadlock-freedom Contract] $\ckey{CDF} \defs \textstyle\rc{\truer}{\exists e @ e \notin \refu'}{\truer}$ \isalink{https://github.com/isabelle-utp/utp-main/blob/90ec1d65d63e91a69fbfeeafe69bd7d67f753a47/theories/sf_rdes/utp_sfrd_fdsem.thy\#L425}
\end{definition}

This reactive contract has a $\truer$ precondition, which by Theorem~\ref{thm:rdesrefine}, means that the precondition
of the implementation contract must also be $\truer$. This is because we must weaken the precondition, and $\truer$ is
the weakest possible reactive condition. Intuitively, this means that to refine $\ckey{CDF}$, a reactive program must
also be free of divergence. The postcondition is also $\truer$, but since we must strengthen the postcondition, any
postcondition for the implementation is admitted. The pericondition contains the main specification formula; it states
that in every quiescent observation there must be an event which is not being refused. In other words, only programs
that do not admit the observation $\refu' = \textit{Event}$ are deadlock-free.

We can show, for example, that $a \then \Skip$ is deadlock-free:

\begin{example}[$\ckey{CDF} \refinedby a \then \Skip$]
  Since the precondition of $a \then Skip$ is $\truer$, it suffices to consider the pericondition, and show that the
  following refinement holds: $$(\exists e @ e \notin \refu') \refinedby a \notin \refu' \land \trace = \langle\rangle$$
  Recall that refinement is reverse implication. Therefore, we need to show that
  $a \notin \refu' \implies (\exists e @ e \notin \refu')$, which straightforwardly holds when we set $e =
  a$. Consequently, we have proved deadlock-freedom.

  Conversely, we cannot show that $\Stop$ is deadlock-free, because its pericondition includes
  $\refu' \subseteq \textit{Event}$, which allows the possibility of refusing everything. \qed
\end{example}

\subsection{Calculational Laws}

Though language-specific operators like those above for \Circus can be expressed, many core contract operators can be
introduced generically. We can therefore develop a large body of laws for calculating contracts that do not depend on a
particular semantic model, but can be instantiated with any trace algebra $\tset$, and additional observational
variables. We begin by denoting some basic reactive operators for this generic theory.
\begin{theorem}[Reactive Design Core Operators] \label{thm:rdescore} \isalink{https://github.com/isabelle-utp/utp-main/blob/90ec1d65d63e91a69fbfeeafe69bd7d67f753a47/theories/rea_designs/utp_rdes_triples.thy\#L492}
\begin{align*}
  \IIsrd &~=~ \rc{\truer}{\false}{\IIr} \\[.3ex]
  \assignsR{\sigma} &~=~ \rc{\truer}{\false}{\assignsr{\sigma}} \\[.3ex]
  \Miracle &~=~ \rc{\truer}{\false}{\false} \\[.3ex]
  \Chaos &~=~ \rc{\false}{\false}{\false}
\end{align*}
\end{theorem}
\noindent Operator $\IIsrd$ is reactive design identity. It has a true precondition, and a false pericondition,
indicating that it has no intermediate states and so is essentially instantaneous. The postcondition defines that it
contributes nothing to the trace, and simply identifies the before and after states. Since the alphabet at this point is
open, by using $\IIr$ as a postcondition, we also add the conjunct $r' = r$ which is shorthand for saying all additional
variables are unchanged. This distinguishes $\IIsrd$ from the \Circus-specific $\Skip$ operator from
Example~\ref{ex:termdead}, which leaves $\refu'$ unconstrained.

Operator $\assignsR{\sigma}$ is a generalised assignment, again similar to Back's update action~\cite{Back1998}, where
$\sigma : \Sigma \to \Sigma$ is a function on the state space. Its postcondition defines an update of the state by
applying $\sigma$ to it using the reactive relational assignment. The more specific assignment $\assignR{x}{v}$ can be
expressed as $\assignsR{\{x \mapsto v\}}$. The generalised assignment also enables us to easily define multiple-variable
assignment constructs.

$\Miracle$ is the miraculous reactive design. It has a true precondition, but has no intermediate or final states, and
thus is effectively impossible to execute. It is the top element of the refinement lattice:
\begin{theorem} $\rc{P_1}{P_2}{P_3} \refinedby \Miracle$ \isalink{https://github.com/isabelle-utp/utp-main/blob/90ec1d65d63e91a69fbfeeafe69bd7d67f753a47/theories/rea_designs/utp_rdes_healths.thy\#L545}
\end{theorem}
\noindent This follows, by Theorem~\ref{thm:rdesrefine}, since the precondition $\truer$ is the least reactive
condition, and $\false$ is the greatest reactive relation. $\Chaos$, in contrast to $\Miracle$, is the contract with an
unsatisfiable precondition and thus always yields a program error. It is the bottom of the refinement lattice, and is
the least deterministic contract:
\begin{theorem} $\Chaos \refinedby \rc{P_1}{P_2}{P_3}$ \isalink{https://github.com/isabelle-utp/utp-main/blob/90ec1d65d63e91a69fbfeeafe69bd7d67f753a47/theories/rea_designs/utp_rdes_healths.thy\#L545}
\end{theorem}
$\Chaos$ can be used to identify interactions that are erroneous, and thus the context should avoid them, as illustrated by
the following example.
\begin{example}[Divergent Process] \label{ex:divproc}
  \begin{align*}
    &a \then \Chaos \extchoice b \then \Skip = \\[.5ex]
    &\qquad \rc{\negr (\langle a \rangle \le \trace)}{\trace = \langle\rangle \land a \notin ref' \land b \notin ref'}{\trace = \langle b \rangle \land st' = st}
  \end{align*}

  \noindent Here, we have utilised Examples~\ref{ex:prefix} and \ref{ex:extchoice}, together with
  Definition~\ref{thm:rdescore} to calculate the composite contract. This \Circus action allows either an $a$ or $b$
  event, but if the environment chooses $a$ then it diverges. The precondition therefore defines the assumption that the
  environment does not extend the trace by $a$, using a reactive condition of the form illustrated in
  Example~\ref{ex:constrprf}. If the program performs $a$, then the behaviour is unpredictable. The pericondition states
  that the trace has not yet been extended, and the action does not refuse $a$ or $b$. However, though it is not
  refused, $a$ can never lead to a terminating state as defined in the postcondition, which specifies that the trace is
  extended by $b$ and leaves the state unchanged. \qed
\end{example}
Contracts can also be constructed using the programming and specification operators of UTP's relational calculus. This
effectively means that relational laws of programming can be directly imported for use in proofs about contracts. We
have proved a number of theorems that show the results of composing contracts.  \allowdisplaybreaks
\begin{theorem}[Reactive Design Compositions] \label{thm:rdes-comp} \isalink{https://github.com/isabelle-utp/utp-main/blob/90ec1d65d63e91a69fbfeeafe69bd7d67f753a47/theories/rea_designs/utp_rdes_triples.thy\#L983}
\begin{align}
  \rc{P_1}{P_2}{P_3} \intchoice \rc{Q_1}{Q_2}{Q_3} &~=~ \rc{P_1 \land Q_1}{P_2 \lor Q_2}{P_3 \lor Q_3} \\[2.5ex]
  \bigsqcap_{i\in I} \rc{P_1(i)}{P_2(i)}{P_3(i)} &~=~ \rc{\bigwedge_{i \in I} P_1(i)}{\bigvee_{i \in I} P_2(i)}{\bigvee_{i \in I} P_3(i)} \label{thm:rdes-comp-inf} \\[2.5ex]
  \rc{P_1}{P_2}{P_3} \sqcup \rc{Q_1}{Q_2}{Q_3} &~=~ 
     \rc{
         \begin{array}{c}
           P_1 \ \\ 
           \lor Q_1
         \end{array}
         }{
         \begin{array}{c}
           P_1 \rimplies P_2 \land \!\! \\ 
           Q_1 \rimplies Q_2
         \end{array}
         }{
         \begin{array}{c}
           P_1 \rimplies P_3 \land \\ 
           Q_1 \rimplies Q_3
         \end{array}} \\[2.5ex]  
  \conditional{\rc{P_1}{P_2}{P_3}}{b}{\rc{Q_1}{Q_2}{Q_3}} &~=~ 
     \rc{\begin{array}{c}
          P_1 \\
          \infixIf b \infixElse \\
          Q_1
         \end{array}
     }{\begin{array}{c}
          P_2 \\
          \infixIf b \infixElse \\
          Q_2
       \end{array}
     }{\begin{array}{c}
          P_3 \\
          \infixIf b \infixElse \\
          Q_3
       \end{array}
     } \\[2.5ex]
  \rc{P_1}{P_2}{P_3} \relsemi \rc{Q_1}{Q_2}{Q_3} &~=~ 
     \rc{\begin{array}{c}
           P_1 \land \\ 
           P_3 \wpR Q_1
         \end{array}
        }{\begin{array}{c}
           P_2 \lor \\
           P_3 \relsemi Q_2
          \end{array}
        }{P_3 \relsemi Q_3} \label{thm:rdes-seq} \\[2.5ex]
  \rc{P}{Q}{R}^{n+1} &~=~ \rc{\bigwedge_{i \le n} \left(R^i \wpR P\right)}{\bigvee_{i \le n} R^i \relsemi Q}{R^{n+1}} \label{thm:rdes-comp-pow}
\end{align}
\end{theorem}
\noindent These are theorems, rather than definitions, since they formulate the semantics of contracts that are composed
by the UTP operators defined previously in Definition~\ref{def:utp-prog}. We do not need to redefine them for our theory
of reactive designs, but rather prove laws that show how to calculate composite contracts. This is a key contribution of
our work, since it means the existing theorems of operators like $\relsemi$ and $\sqcap$ can be directly imported into
our theory of reactive designs, and applied to algebraic reasoning.

The internal choice of two contracts, ($P \intchoice Q$), yields a contract that assumes both preconditions
hold, and yields the combined intermediate and final states by disjunction. The preconditions are conjoined since the
choice is non-deterministic, and thus there must be no possibility of divergence in any of the possible
branches. Internal choice can, alternatively, be viewed as a disjunction operator for contracts similar to that
in~\cite{Benveniste2007}. Similarly, an internal choice over a set of basic actions indexed by a set $I$ conjoins all
the preconditions, and disjoins the peri- and postconditions. Dual to disjunction, the conjunction of two contracts
($P \sqcup Q$) requires that one of the preconditions holds, and takes the conjunction of the corresponding intermediate
and final states. The conditional $\conditional{P}{b}{Q}$, where $b$ is a predicate on $st$ alone, can be distributed
through the pre-, peri-, and postconditions of the respective reactive designs.

Sequential composition $P \relsemi Q$, where $P = \rc{P_1}{P_2}{P_3}$ and $Q = \intchoice \rc{Q_1}{Q_2}{Q_3}$, is a
little more involved. The combined precondition conjoins the precondition of $P$ with a predicate requiring that the
postcondition of $P$ does not violate the precondition of $Q$. The latter is specified using the reactive weakest
precondition operator, $P \wpR Q$. The pericondition states that either $P$ is in an intermediate state, and thus $P_2$
holds, or else $Q$ is in intermediate state, $P$ having terminated, and thus $P_3 \relsemi Q_2$ holds. Finally, the
postcondition states that both $P$ and $Q$ have terminated, that is, $P_3 \relsemi Q_3$.

As a corollary, we prove the law for finite iteration of a reactive design, $P^{n+1}$, assuming at least one execution,
that is, for $\rc{P}{Q}{R} \relsemi \rc{P}{Q}{R} \relsemi \cdots \relsemi \rc{P}{Q}{R}$. This law can be applied to
calculate the contract for a recursive reactive program. The precondition requires that after $i \le n$ iterations of
the postcondition $R$, the precondition $P$ is not violated. The pericondition states that postcondition $R$ has been
established a number of times $i \le n$, followed by the pericondition $Q$ holding. In other words, one of the
iterations is still in an intermediate state. Finally, the overall postcondition states that $R$ has been established
$n + 1$ times.

Finally we present a law for calculating the contract of a tail-recursive program of the form
$\thlfp{R}\, X @ P \relsemi X$, where $\thlfp{R}$ is the weakest fixed-point operator, which allows us to formulate
iterative contracts. This is subject to $P$ being a productive~\cite{Coquand1993} contract, that is, one that extends
the trace when it terminates. 

\begin{definition} \label{def:productive} $\rc{P_1}{P_2}{P_3}$ is said to be productive if $P_3$ is a
  fixed-point of $\ckey{R4}(P) \defs P \land tr < tr'$, that is if we establish termination then it is necessary that
  the trace strictly increases. \isalink{https://github.com/isabelle-utp/utp-main/blob/90ec1d65d63e91a69fbfeeafe69bd7d67f753a47/theories/rea_designs/utp_rdes_productive.thy\#L12}
\end{definition}

For example, $a \then \Skip$ is productive because it always produces an $a$ event upon termination. Its postcondition
is $\healthy{R4}$ healthy because it contains the conjunct $\trace = \langle a \rangle$, which strictly increases the
trace ($tr < tr'$). On the other hand, $\assignsR{\sigma}$ is not productive because it contributes no events to the
trace. Productivity is related to, but not the same as the common notion of ``guardedness''~\cite{Hoare&98}, which, as
explained in Section~\ref{sec:recurse}, applies to a function on contracts rather than a contract itself. If a
contract's postcondition is productive, then we have the following theorem.

\begin{theorem}[Recursive Reactive Design] \label{thm:rdes-rec} 

  \noindent If $R$ is $\healthy{R4}$ healthy, then
  \begin{align*}
  \thlfp{R}\, X @ \rc{P}{Q}{R} \relsemi X &= \rc{\bigwedge_{i \in \nat} \left(R^i \wpR P\right)}{\bigvee_{i \in \nat} R^i \relsemi Q}{\false}
\end{align*}
\end{theorem}

\noindent Such a recursive contract has a false postcondition, since it does not terminate. The precondition requires
that, no matter how many times postcondition $R$ is established, it does not violate the contract's precondition
$P$. The pericondition is where the main behaviour of the contract is specified. It states $R$ is executed some number
of times, and then the pericondition holds. In other words, the contract has executed its body and terminated into a
final state of the body several times, but then finally the contract always lands in an intermediate state, since it
does not terminate itself.

We now give some of the algebraic laws of reactive design contracts.
\begin{theorem}[Reactive Design Laws]
  \begin{align*}
    \Miracle \intchoice \rc{P_1}{P_2}{P_3} &~~=~~ \rc{P_1}{P_2}{P_3} \tag{RD1} \label{law:RD1} \\
    \Chaos \intchoice \rc{P_1}{P_2}{P_3} &~~=~~ \Chaos \tag{RD2} \label{law:RD2} \\
    \IIsrd \relsemi \rc{P_1}{P_2}{P_3} &~~=~~ \rc{P_1}{P_2}{P_3} \tag{RD3} \label{law:RD3} \\
    \rc{P_1}{P_2}{P_3} \relsemi \IIsrd &~~=~~ \rc{P_1}{P_2}{P_3} \tag{RD4} \label{law:RD4} \\
    \rc{P_1}{P_2}{\false} \relsemi \rc{Q_1}{Q_2}{Q_3} &~~=~~ \rc{P_1}{P_2}{\false} \tag{RD5} \label{law:RD5} \\
    \Miracle \relsemi \rc{P_1}{P_2}{P_3} &~~=~~ \Miracle \tag{RD6} \label{law:RD6} \\
    \Chaos \relsemi \rc{P_1}{P_2}{P_3} &~~=~~ \Chaos \tag{RD7} \label{law:RD7} \\
    \rc{\false}{P_2}{P_3} &~~=~~ \Chaos \tag{RD8} \label{law:RD8} \\
    \rc{P_1}{P_2}{P_3} \relsemi \Miracle &~~=~~ \rc{P_1}{P_2}{\false} \tag{RD9} \label{law:RD9} \\
    \rc{P_1}{P_2}{P_3} \relsemi \Chaos   &~~=~~ \rc{P_1 \land (P_3 \wpR \false)}{P_2}{\false} \tag{RD10} \label{law:RD10}
  \end{align*}
\end{theorem}
\noindent All of these laws can be proved by calculation using the definitions in Theorem~\ref{thm:rdescore} and laws in
Theorem~\ref{thm:rdes-comp}. \ref{law:RD1} establishes that a choice between a $\Miracle$ and $P$ yields $P$, since
$\Miracle$ is the top of the lattice. Similarly, \ref{law:RD2} establishes that a choice between $\Chaos$ and $P$ yields
$\Chaos$. The reactive skip is a left and right identity for any contract $P$, as this is stated by \ref{law:RD3} and
\ref{law:RD4}, respectively.

Law \ref{law:RD5} states that any non-terminating contract -- that is where the postcondition is $\false$ -- is a left
zero for sequential composition, as clearly then the successor is unreachable. Thus, in particular $\Miracle$ and
$\Chaos$ are both left zeros for sequential composition, as shown by \ref{law:RD6} and \ref{law:RD7}. Moreover,
\ref{law:RD7} shows that any reactive contract with a $\false$ precondition is $\Chaos$.

Law \ref{law:RD9} is a property first observed in \cite{Woodcock08}: placing a $\Miracle$ after a reactive design
eliminates final states, and yields a non-terminating process. Since it is impossible to reach a miraculous state,
inserting one prunes transitions that lead to it.

Finally, \ref{law:RD10} is a similar law for $\Chaos$, which likewise removes final states. Crucially, however, the
behaviour of $\Chaos$ is not impossible, but simply undesirable or unpredictable. Thus the composition additionally
inserts an assumption $P_3 \wpR \false$, which effectively states that postcondition $P_3$ should not be established,
because otherwise chaos will ensue. This explains Example~\ref{ex:divproc}: the left branch of the choice,
$a \then \Chaos$ is equivalent to $(a \then \Skip) \relsemi \Chaos$. The postcondition of $a \then \Skip$ is
$st' = st \land \trace = \langle a \rangle$.  The occurrence of $\Chaos$ mandates that this postcondition should not be
established, which means that trace extension is negated and added to the assumption, yielding the reactive condition
$\negr (\langle a \rangle \le \trace)$. This important distinction illustrates the difference between $\Miracle$ and
$\Chaos$ -- usually the latter is used to encode behaviour that should be prevented by the environment.

The next theorem gives the laws of reactive assignment.
\begin{theorem}[Reactive Assignment Laws] \label{thm:rea-asn}
  \begin{align*}
    \assignsR{id} &~=~ \IIsrd \tag{RA1} \label{law:RA1} \\
    \assignsR{\sigma} \relsemi \rc{P_1}{P_2}{P_3} &~=~ \rc{\substapp{\sigma}{P_1}}{\substapp{\sigma}{P_2}}{\substapp{\sigma}{P_3}} \tag{RA2} \label{law:RA2} \\
    \assignsR{\sigma} \relsemi \assignsR{\rho} &~=~ \assignsR{\rho \circ \sigma} \tag{RA3} \label{law:RA3} \\
    \assignsR{\sigma} \relsemi \Miracle &~=~ \Miracle \tag{RA4} \label{law:RA4} \\
    \assignsR{\sigma} \relsemi \Chaos &~=~ \Chaos \tag{RA5} \label{law:RA5}
  \end{align*}
\end{theorem}
\noindent Law \ref{law:RA1} establishes that an assignment using the identity function $id$ yields the reactive
skip. \ref{law:RA2} captures the effect of precomposing a reactive contract with an assignment; the assignment function
is applied as a substitution in the pre-, peri-, and postconditions. \ref{law:RA3} states that composition of two
assignments yields a single assignment built by composition of the individual assignment functions. Laws \ref{law:RA4}
and \ref{law:RA5} establish that $\Miracle$ and $\Chaos$ are both right zeros for assignment. This is because they both
remove final states, but, since assignments have no intermediate states, this eliminates all observable behaviours.

\subsection{Parallel Contracts}
\label{sec:parcontract}

\begin{figure}
  \centering\includegraphics[width=.5\linewidth]{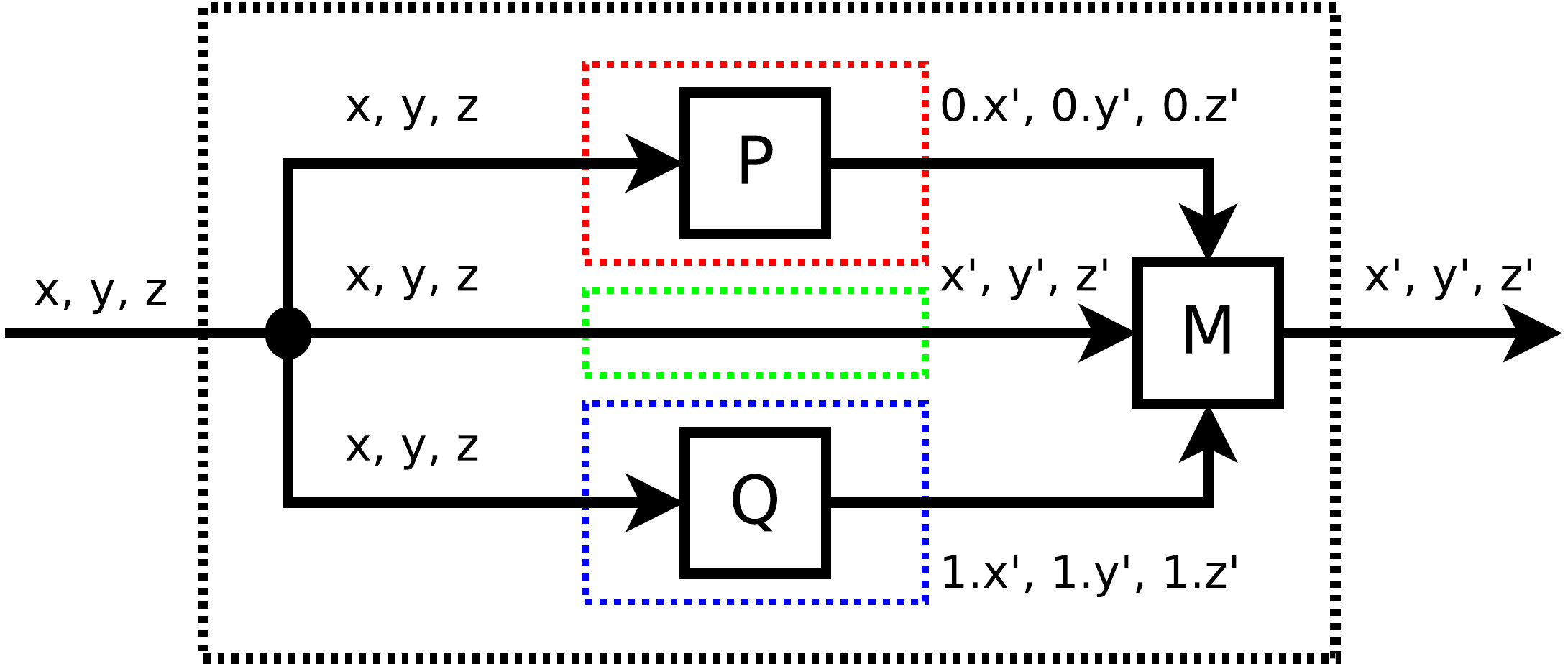}

  \caption{Parallel-by-merge Dataflow}
  \label{fig:pbm}

\end{figure}

The final operator we tackle in this section is parallel composition, written $P \rcpar{M} Q$. Our definition of
parallel composition, elaborated in Section~\ref{sec:parallel}, uses the parallel-by-merge scheme developed as part of
UTP~\cite{Hoare&98}. Since different concurrency schemes are possible for reactive contracts, depending on the
underlying notion of trace, we cannot define a single parallel composition operator and so $P \rcpar{M} Q$ is parametric
over $M$. This is a merge predicate that defines how the state, traces, and any other observational variables should be
merged following execution of $P$ and $Q$.

We illustrate parallel-by-merge in Figure~\ref{fig:pbm}, where we assume the programs act on three variables, $x$, $y$,
and $z$. Parallel-by-merge splits the observation space into three identical segments: one for $P$, one for $Q$, and a
third that is identical to the original input. Relation $M$ then takes the outputs from $P$, $Q$, and the original
inputs, and merges them into a single output.

For reactive designs, the variables we need to merge are the state variables ($\state$), trace ($\trace$), and any
additional observational variables $r$. As an example, we present below the merge predicate for interleaving the events
of two \Circus processes~\cite{Hoare&98,Oliveira2005-PHD}.
\begin{example}[Interleaving Merge] \label{ex:interleave}
  \begin{align*}
    M_c ~~~\defs~~~& \trace \in (0.\trace \interleave_t 1.\trace) \land ref' \subseteq 0.\refu \cap 1.\refu \quad\mbox{where} \\[1ex]
    \langle\rangle \interleave_t xs ~~~=~~~ & xs \interleave_t \langle\rangle ~~~=~~~ \{xs\} \\
    \langle x \rangle \cat xs \interleave_t \langle y \rangle \cat ys ~~~=~~~
                 & \left\{\langle x \rangle \cat zs ~~ \middle| ~~ zs \in (xs \interleave_t \langle y \rangle \cat ys) \right\} ~\cup~ \\
                 & \left\{\langle y \rangle \cat zs ~~ \middle| ~~ zs \in (\langle x \rangle \cat xs \interleave_t ys) \right\} \\[1ex]
    P \interleave Q ~~~\defs~~~ & P \rcpar{M_c} Q
  \end{align*}

  \noindent In the definition of $P \rcpar{M} Q$, the trace and refusal variables are decorated with an index $0$ or $1$
  that determine whether the quantity originates from $P$ or $Q$, respectively. The binary operator $\interleave_t$
  interleaves two traces; it is a recursive function on sequences, returning a set of possible traces. The merge
  predicate, $M_c$, firstly constructs the overall trace $\trace$ as one of all possible interleavings, and secondly
  states that an event is only refused if it is refused by both processes. Our merge predicate leaves the state variable
  $\state$ unspecified as the \Circus interleaving operator hides any internal state, since they can not in general be
  merge without further machinery for shared variables. We then define the interleaving operator $P \interleave Q$ using
  the parametric reactive design parallel composition operator $\rcpar{\cdot}$, which is formally defined in
  Section~\ref{sec:parallel}. \qed
\end{example}
For the purpose of generic laws, we assume that any merge predicate yields well-formed traces, and furthermore is
symmetric. Symmetry in this context means that the merge predicate has no bias towards either of its operands and yields
the same result if they are swapped. This is clearly the case in Example~\ref{ex:interleave}, since both the traces and
refusals are composed by symmetric operators, namely $\interleave_t$ and $\cap$. The following theorem describes the
result of composing two contracts.
\begin{theorem}[Reactive Design Parallel Composition] \label{thm:rdespar} \isalink{https://github.com/isabelle-utp/utp-main/blob/07cb0c256a90bc347289b5f5d202781b536fc640/theories/rea_designs/utp_rdes_parallel.thy\#L769}
  \begin{align*}
    &\rc{P_1}{P_2}{P_3} \rcpar{M} \rc{Q_1}{Q_2}{Q_3} = \\[.5ex]
    & \qquad \rc{\begin{array}{l} 
          (P_1 \rimplies P_2) \wppR{M} Q_1 \land \\
          (P_1 \rimplies P_3) \wppR{M} Q_1 \land \\
          (Q_1 \rimplies Q_2) \wppR{M} P_1 \land \\
          (Q_1 \rimplies Q_3) \wppR{M} P_1
       \end{array}}{\begin{array}{l}
          P_2 \rcmergee{M} Q_2 ~\lor~ \\[1ex]
          P_3 \rcmergee{M} Q_2 ~\lor~ \\[1ex]
          P_2 \rcmergee{M} Q_3
       \end{array}
       }{P_3 \parallel_{\text{\tiny $M$}} Q_3}
  \end{align*}
\end{theorem}
\noindent The precondition is expressed in terms of a reactive condition combinator $\wppR{M}$. This is a form of
weakest rely condition: $A \wppR{M} B$ describes the weakest context in which reactive relation $A$ does not violate
reactive condition $B$. This is necessary because in a composition like $P \rcpar{M} Q$, the reactive processes $P$ and
$Q$ interfere and this can lead to the violation of their preconditions. Thus, the overall precondition of a parallel
composition itself assumes that such interferences do not occur. Interferences cannot occur in
Example~\ref{ex:interleave}, since there is no synchronisation, but, in general, of course it can happen when more
specialised merge predicates are employed. The definition of $\wppR{M}$ is elided for now, as it requires further
elaboration of the UTP theory, but is given in Section~\ref{sec:grd}, Definition~\ref{def:wppR}. It obeys several
theorems that are shown below:

\begin{theorem}[Weakest Rely Laws] \label{thm:wrlaws} \isalink{https://github.com/isabelle-utp/utp-main/blob/07cb0c256a90bc347289b5f5d202781b536fc640/theories/rea_designs/utp_rdes_parallel.thy\#L720}
  $$ \false \wppR{M} P = \truer \qquad
    P \wppR{M} \truer = \truer \qquad 
    \left(\bigsqcap_{i \in I} ~ P(i)\right) \wppR{M} Q = \left(\bigwedge_{i \in I} ~ (P(i) \wppR{M} Q)\right)
  $$
\end{theorem}

The laws show, respectively, that (1) a miraculous reactive relation satisfies any precondition, (2) any reactive
relation satisfies a true precondition, and (3) the weakest rely condition of a disjunction of relations is the
conjunction of their weakest rely conditions.

The parallel composition precondition in Theorem~\ref{thm:rdespar} states that the respective preconditions of the two
contracts, $P_1$ and $Q_1$, are not violated, neither by the opposing periconditions -- respectively $Q_2$ and $P_2$ --
under their respective preconditions, or the opposing postconditions -- respectively $Q_3$ and $P_3$. The pericondition
of the parallel reactive contract is written in terms of operator $P \rcmergee{M} Q$ that merges the traces but not the
states, since this is concealed in intermediate observations. A parallel process is in an intermediate state if either
of the composed processes is. Finally, the postcondition merges the final trace and state of each process, directly using
the parallel by merge operator ($\parallel_M$). We formally define all these operators in Section~\ref{sec:parallel}.

We can use Theorems~\ref{thm:rdespar} and \ref{thm:wrlaws} to prove the following theorem; for illustration we elaborate
the complete calculation.
\begin{theorem} \label{thm:miracle-anhil} $\Miracle \rcpar{M} P = \Miracle$ \isalink{https://github.com/isabelle-utp/utp-main/blob/07cb0c256a90bc347289b5f5d202781b536fc640/theories/rea_designs/utp_rdes_parallel.thy\#L813} \end{theorem}
\begin{proof}
  \begin{align*}
    \Miracle \rcpar{M} P 
    &= \rc{\truer}{\false}{\false} \rcpar{M} \rc{P_1}{P_2}{P_3} \\
    &= \rc{\begin{array}{l} 
        (\truer \rimplies \false) \wppR{M} P_1 \land \\
        (\truer \rimplies \false) \wppR{M} P_1 \land \\
        (P_1 \rimplies P_2) \wppR{M} \truer \land \\
        (P_1 \rimplies P_3) \wppR{M} \truer
      \end{array}}{
      \begin{array}{l}
        \false \rcmergee{M} P_2 ~\lor~ \\
        \false \rcmergee{M} P_2 ~\lor~ \\ 
        \false \rcmergee{M} P_3
      \end{array}
    }{\false \parallel_{\text{\tiny M}} P_3} & [\ref{thm:rdespar}] \\
    &= \rc{\begin{array}{l} 
        \false \wppR{M} P_1 \land \false \wppR{M} P_1 \land \\
        \truer \land \truer
      \end{array}}{
      \begin{array}{l}
        \false ~\lor~  \false ~\lor~ \false
      \end{array}
    }{\false} & [\ref{thm:wrlaws}] \\
    &= \rc{\begin{array}{l} 
        \truer \land \truer
      \end{array}}{
      \begin{array}{l}
        \false
      \end{array} 
    }{\false} & [\ref{thm:wrlaws}] \\
    &= \rc{\truer}{\false}{\false} \\
    &= \Miracle & \qedhere 
  \end{align*} 
\end{proof}
\noindent This law shows that composition of any predicate with a miracle always yields a miracle, regardless of the
merge predicate. The intuitive property that $\Chaos$ is similarly an annihilator depends on the form of merge
predicate, and so this cannot be proved in general. We perform a further example calculation using the interleaving
operator from Example~\ref{ex:interleave}.

\begin{example}[Interleaving Calculation] \isalink{https://github.com/isabelle-utp/utp-main/blob/90ec1d65d63e91a69fbfeeafe69bd7d67f753a47/tutorial/utp_csp_ex.thy}
\begin{align*}
  & \quad~~ a \then \Skip \interleave b \then \Stop \\[.5ex]
  & = \rcs{a \notin \refu' \land \trace = \snil}{\state' = \state \land \trace = \langle a \rangle} \interleave
    \rcs{a \notin \refu' \land \trace = \snil \lor \trace = \langle b \rangle}{\false} \\[.5ex]
  & = \rcs{
    \begin{array}{l}
      (a \notin \refu' \land \trace = \snil) \rcmergee{M_C} (b \notin \refu' \land \trace = \snil) \lor \\
      (a \notin \refu' \land \trace = \snil) \rcmergee{M_C} (\trace = \langle b \rangle) \lor \\
      (\state' = \state \land \trace = \langle a \rangle) \rcmergee{M_C} (b \notin \refu' \land \trace = \snil) \lor \\
      (\state' = \state \land \trace = \langle a \rangle) \rcmergee{M_C} (\trace = \langle b \rangle)
    \end{array}}{(\state' = \state \land \trace = \langle a \rangle) \parallel_{\text{\tiny $M$}} \false} \\
  & = \rcs{
    \begin{array}{l}
      (a \notin \refu' \land b \notin \refu' \land \trace = \snil) \lor \\
      (a \notin \refu' \land \trace \in \{\snil, \langle b \rangle\}) \lor \\
      (b \notin \refu' \land \trace \in \{\snil, \langle a \rangle\}) \lor \\
      (\trace = \langle a, b \rangle)
    \end{array}}{\false} \\
  & = \rcs{\conditional{(a \notin \refu' \land b \notin \refu')}{\trace = \snil}{\left(\begin{array}{l} a \notin \refu' \land \trace = \langle b \rangle \\ \lor b \notin \refu' \land \trace = \langle a \rangle \\ \lor \trace = \langle a, b \rangle} \end{array}\right)}{\false} \\
  & = a \then b \then \Stop \extchoice b \then a \then \Stop
\end{align*}
\noindent We use the abbreviation $\rcs{P_2}{P_3}$ for a contract with a $\truer$ precondition. In the first step, we
calculate the meaning of the two sequential processes using Examples~\ref{ex:prefix} and \ref{ex:termdead}, with Theorem
\ref{thm:rdes-comp}. We then employ Theorem~\ref{thm:rdespar} to expand out the overall parallel reactive
contract. There is no possibility of divergence, so the preconditions remain trivial. For the pericondition we need to
merge each quiescent observation with every opposing quiescent or final observation. The result is the four conjuncts
displayed. In the postcondition, we have to merge the two overall postconditions. The overall postcondition becomes
$\false$, because $\Stop$ prevents the overall operator from successfully terminating. In the pericondition, we
calculate the four merged quiescent observations. These characterise the states when (1) no event has yet occurred, and
we are accepting $a$ or $b$; (2) either $\snil$, or else $\langle b \rangle$ has occurred and $a$ remains enabled, which
includes the situation when the right hand action has performed a transition; (3) the symmetric case to (2) with $a$
potentially having occurred; and (4) both events have occurred and no event is enabled: the refusal set is
unconstrained. In common with the semantics of CSP\cite{Roscoe2005}, we observe that the set of possible refusals is
downward closed under $\subseteq$ for every trace combination. We can then in fact show that this contract is
equivalent to an external choice using the definition given in Example~\ref{ex:extchoice}. \qed
\end{example}

This completes our exposition of the calculational laws for our reactive contract theory. All the theorems presented
have been mechanically checked, as we discuss in Section~\ref{sec:mech}. In the next section we illustrate their use in
verification.

%% file: sec/mondex.tex
In this section, we exemplify the use of reactive contracts to verify properties of a small cash-card system described
in \Circus, using the verification procedure outlined in Section~\ref{sec:contref}. We will explicitly calculate the
implementation contracts, and show how these are then verified. Though the corresponding calculations are complex, the
crucial detail is that they symbolically characterise reactive programs with potentially infinite state, and can be
produced automatically in Isabelle/UTP.

In this example, cards can independently perform transfers to one another, provided sufficient balance exists. A key
requirement is that there is no loss or increase of the value shared across the cards. For simplicity, we will construct
a purely sequential specification of this system, and show how these properties can be discharged. We use a modified
version of the model described in \cite{Woodcock2014b}, which describes a network of a number of cards, each of which is
identified by a natural number $\nat$. Monetary amounts are represented as integers, $\num$, so that we can additionally
represent negative balances.

The central process has a single state variable $accts : \nat \pfun \num$, a partial function that represents the set of
accounts: the balance on each card. We also introduce three channels:
\begin{itemize}
  \item $pay : \nat \times \nat \times \num$, to initiate a transfer request of a given value between two cards;
  \item $reject : \nat$, to indicate rejection of a transfer from the given card identifier; 
  \item $accept : \nat$, to indicate acceptance.
\end{itemize}
In order to update a particular account stored in $accts$, we need a form of assignment that applies to a single
entry. We therefore introduce the following indexed assignment operator:
\begin{definition}[Indexed Assignment] \label{def:iassign} \isalink{https://github.com/isabelle-utp/utp-main/blob/90ec1d65d63e91a69fbfeeafe69bd7d67f753a47/theories/sf_rdes/utp_sfrd_prog.thy\#L331}
  $$\assigniC{x(i)}{v} ~~\defs~~ \rc{i \in \dom(\state.x)}{\false}{\state' = \state(x \mapsto x(i \mapsto v)) \land \trace = \snil}$$
\end{definition}
Indexed assignment $\assigniC{x(i)}{v}$, as employed by the $Pay$ action, is unlike regular assignment in that it is a
partial operator and can only be executed when the collection $x$ has the index $i$ defined. Consequently, it must be
guarded by an assumption $i \in \dom(\state.x)$, which states that the index $i : A$ is in the domain of variable
$x : A \pfun B$.

We give indexed assignment a denotational semantics using a \Circus contract, thus extending the available operators,
which also illustrates the extensibility of our approach. The construct has similar semantics to regular assignment, but
has the precondition that the given collection index must exist. The postcondition states that the collection state
variable $x$ is updated so that the index $i$ maps to $v$. If the precondition is violated, then the result is
divergence as the following calculation demonstrates:
\begin{example}[Divergent Indexed Assignment]
\begin{align*}
  & \quad \assignR{x}{\emptyset} \relsemi \assigniC{x(i)}{v} \\
  &= \!\!\!\begin{array}{l}
       \rc{\truer}{\false}{\state' = \state(x \mapsto \emptyset) \land \trace = \snil} \relsemi \\
       \rc{i \in \dom(\state.x)}{\false}{\state' = \state(x \mapsto x(i \mapsto v)) \land \trace = \snil}
     \end{array} & [\ref{thm:rdescore}, \ref{def:iassign}] \\
  &= \rc{\truer \land (\state' = \state(x \mapsto \emptyset) \land \trace = \snil) \wpR (i \in \dom(\state.x))}{\false}{\cdots} & [\ref{thm:rdes-seq}] \\
  &= \rc{i \in \dom(\emptyset)}{\false}{\cdots} & [\text{relational calculus}] \\
  &= \Chaos & [\ref{law:RD8}]
\end{align*}
If we assign $\emptyset$, the empty partial function, to $x$ and then attempt to manipulate the $i$th element of $x$,
the result is always $\Chaos$. \qed
\end{example}

If we use an indexed assignment within a reactive program, then it is necessary to show that the applied indices are
always in scope. We can now define the key actions for the cash card system.
\begin{definition}[Card System] \isalink{https://github.com/isabelle-utp/utp-main/blob/90ec1d65d63e91a69fbfeeafe69bd7d67f753a47/tutorial/utp_csp_mini_mondex.thy\#L38}
  \begin{align*}
    Pay(i, j : \nat, n : \num) ~~\defs~~~~ & pay.i.j.n \then \\
                       & ~~
                         \left(\begin{array}{l}
                           reject.i \then \Skip \\[1ex]
                           \quad \infixIf i = j \lor i \notin \dom(accts) \lor n \le 0 \lor n > accts(i) \infixElse \\[1ex]
                           \left(\begin{array}{l}
                           \assigniC{accts(i)}{accts(i) - n \relsemi} \\
                           \assigniC{accts(j)}{accts(j) + n \relsemi} \\
                           accept.i \then \Skip
                           \end{array}\right)
                         \end{array}\right) \\
    PaySet(cs : \power \nat) ~~\defs~~~~ & \{(i : \nat,j : \nat,n : \num) | i \in cs \land j \in cs \land i \neq j\} \\
    SomePay(cs : \power \nat) ~~\defs~~~~ & \bigsqcap (i,j,n) \in PaySet(cs) @ Pay(i,j,n) \\
    Cycle(cs : \power \nat) ~~\defs~~~~ & \mu X @ SomePay(cs) \relsemi X \\
    System ~~\defs~~~~ & \assigniC{accts}{\langle 100, 100, 100, 100, 100 \rangle} \relsemi Cycle(\{0..4\})
  \end{align*}
\end{definition}

\noindent Action $Pay(i,j,n)$ defines the protocol to execute a payment request between cards $i$ and $j$ of amount
$n$. If $i = j$, and thus the cards are the same, or card $i$ does not exist, then the transfer is rejected, by offering
$reject.i$ and then terminating. Likewise, if there is insufficient balance or a negative transfer is requested, then
the transfer is also rejected. If the transfer can be performed, then two indexed assignments lower the balance of card
$i$ and raise the balance of card $j$, respectively. Finally, the $accept.i$ event is offered and then the action
terminates.

The main behaviour of the system is described by the $Cycle$ action, which takes the set of card identifiers
$cs : \power \nat$ as a parameter. It iterates the action $SomePay(cs)$, which consists of an internal choice over all
possible payments between all possible cards, $PaySet(cs)$. The behaviour of an example card system is described by the
action $System$ that creates 5 cards, each with a balance of 100, and then begins the cycle.

In order to verify properties of the processes, we first need to calculate the reactive design contract of the
system. For the purposes of illustration, we focus on the $Pay$ action. The other contracts can be calculated in terms
of the $Pay$ contract using the laws presented in Section~\ref{sec:contracts}. The following theorem provides the result
of the calculation. Like all other theorems in this paper, it is proved, not by hand, but by use of our tactics in
Isabelle/UTP. %

\begin{theorem}[Pay Action Contract Calculation] \label{thm:mondexcontract} \isalink{https://github.com/isabelle-utp/utp-main/blob/90ec1d65d63e91a69fbfeeafe69bd7d67f753a47/tutorial/utp_csp_mini_mondex.thy\#L95}
  \vspace{1ex}

  $Pay(i,j,n) = \rc{P_1}{P_2}{P_3}$ where

  \vspace{-3ex}
  \begin{align*}
    P_1 =~~~~& \langle pay.i.j.n \rangle \le \trace \land i \neq j \land i \notin \dom(accts) \land 0 < n \land n < accts(i) \\ 
    \implies~~~~& j \in \dom(accts) \\[1.5ex]
    P_2 =~~~~& \trace = \langle\rangle \land pay.i.j.n \notin \refu' \\
        \lor~~~~ & \trace = \langle pay.i.j.n \rangle \land (i = j \lor n < 0 \lor n > accts(i)) \land reject.i \notin \refu' \\
        \lor~~~~ & \trace = \langle pay.i.j.n \rangle \land i \neq j \land n < 0 \land n > accts(i) \land accept.i \notin \refu' \\[1.5ex]
    P_3 =~~~~& \trace = \langle pay.i.j.n, reject.i \rangle \land (i = j \lor n < 0 \lor n > accts(i)) \land \state' = \state \\
      \lor~~~~ & \trace = \langle pay.i.j.n, accept.i \rangle \land i \neq j \land n < 0 \land n > accts(i)  \\
        &\qquad \land accts := accts(i \mapsto accts(i) - n, j \mapsto accts(j) + n)
  \end{align*}
\end{theorem}

\noindent Theorem~\ref{thm:mondexcontract} gives the result of the calculation of the precondition $P_1$, pericondition
$P_2$, and postcondition $P_3$ of the implementation contract for the $Pay$ action. Intuitively, this complex contract
symbolically characterises the potentially infinite state space and possible transitions that the $Pay$ action can make,
depending on the initial state. The precondition specifies the circumstances under which behaviour is predictable, the
pericondition specifies the possible quiescent observation, and the postcondition specifies the terminating observations.

Precondition $P_1$ requires that, if the $pay.i.j.n$ event occurs, that is, it is at the head of the trace, and the
conditions for a valid transfer are all satisfied by the state, then it must be the case that card $j$ exists in the
state space. This precondition arises directly from the indexed assignment: its violation leads to unpredictable
behaviour. Thus, if the transfer is not offered or the payment is not valid then this assignment is not reached, and so
the precondition precisely identifies the state in which divergence is possible. We alternatively could remove this
precondition altogether by altering the definition of $Pay$ so that if $j \notin \dom(accts)$ then also a $reject$ event
is issued. However, for illustration purposes we leave the precondition in place. This precondition is a reactive
condition because it only restricts a prefix of the trace to the left of the implication and only refers to initial
state variables.

The pericondition, $P_2$, specifies the three possible intermediate observations for the action. Firstly, we have the
scenario in which nothing has happened yet, so the trace is empty and the $pay.i.j.n$ is being offered -- it is not
being refused. Secondly, it is possible that the transfer request occurred, and so the trace has a singleton event, but
one of the conditions for a valid transfer is violated, and thus the $reject.i$ event is being offered. Thirdly, we can
have that a valid transfer request occurred and so the $accept.i$ event is being offered. At this point the state update
has happened internally, but it cannot yet be observed as the action is still in an intermediate state.

The postcondition, $P_3$, specifies two possible final states, one for an invalid transfer request, in which case the
state remains the same, and one for a valid transfer, in which case the two balances are updated. In both cases the
trace is updated with two events, and no refusals are recorded since we have terminated.

The next step in verification is to specify some holistic properties of our system that we would like to show. We chose
three properties: (1) there is no increase or decrease in the overall balance across cards, (2) no overdrafts on the
card balances are permitted, and (3) if a valid transfer request is made it must be executed. It is not difficult to see
that these properties hold, but the purpose is to show how they can specified and verified using reactive contracts.

Each of these properties is assigned a contract which $Pay(i,j,n)$ must refine. We therefore demonstrate the three
properties as three corresponding theorems, which have been discharged in Isabelle/UTP using the \textsf{rdes-refine}
tactic (see Section~\ref{sec:mech}), which employs Theorem~\ref{thm:rdesrefine} and the contract calculation laws. For
the purpose of illustration, we give high-level informal proofs that correspond to the mechanised proofs.

\begin{theorem}[No Increase or Decrease in Value] A payment between card $i$ and $j$, where $\{i, j\} \subseteq cs$ and
  $i \neq j$, does not lead to an overall change in balance for the system of cards. Formally, as a reactive
  contract: $$\rc{\dom(accts) = cs}{\truer}{sum(accts) = sum(accts')} \refinedby Pay(i,j,n)$$ where $sum$ is a function
  that sums up the range of the given partial function. We set the pericondition $\truer$ as we do not need to constrain
  quiescent states in this specification. \isalink{https://github.com/isabelle-utp/utp-main/blob/90ec1d65d63e91a69fbfeeafe69bd7d67f753a47/tutorial/utp_csp_mini_mondex.thy\#L149}
\end{theorem}

\begin{proof}
  We first apply Theorem~\ref{thm:rdesrefine} to split the refinement into three proof obligations. We need to show that
  the precondition is weakened, and the peri- and postconditions are strengthened, which we tackle one at a time:
  \begin{enumerate}
    \item $\dom(accts) = cs \implies P_1$
    \item $P_2 \land \dom(accts) = cs \implies \truer$
    \item $P_3 \land \dom(accts) = cs \implies sum(accts) = sum(accts')$
  \end{enumerate}
  Case (1) follows because $i, j \in cs$ and since the domain of $accts$ is $cs$, then clearly both $i, j \in
  \dom(accts)$. Case (2) follows trivially. Case (3) requires that we consider both cases of postcondition $P_3$. If
  the payment request is invalid, then $accts' = accts$ since $\state' = \state$, and thus clearly $sum(accts) =
  sum(accts')$. If the payment request is valid, then we need to show that $sum(accts) = sum(accts(i \mapsto accts(i) - n, j \mapsto
  accts(j) + n))$. This is equal to $sum(accts) - n + n = sum(accts)$, and so we are done.
\end{proof}

\begin{theorem}[No overdrafts] No card $i$ is permitted to have a negative balance: \isalink{https://github.com/isabelle-utp/utp-main/blob/90ec1d65d63e91a69fbfeeafe69bd7d67f753a47/tutorial/utp_csp_mini_mondex.thy\#L192} $$\rc{\dom(accts) = cs}{\truer}{\forall k \in cs @ accts(k) \ge 0 \implies accts'(k) \ge 0} \refinedby Pay(i,j,n)$$

\end{theorem}

\begin{proof}
  The argument for the pre and periconditions is the same as in the previous theorem.  For the postcondition we need to
  show that after a valid payment the balance of any valid card $k$ is not less than 0, that is, $accts(k) \ge 0$. The
  key property to prove here is $(accts(i \mapsto accts(i) - n, j \mapsto accts(j) + n))(k) \ge 0$. We can do this by
  case analysis: $k = i$, $k = j$, and $k \neq i \land k \neq j$. In the former two cases, the fact that the transfer
  happens indicates that the balances following must be no less than zero. In the final case, the balance remains the same,
  and so we are done.
\end{proof}

\begin{theorem}[Transfer Acceptance] If a payment is initiated and we have enough money in the account, then the
  transfer is not rejected: \isalink{https://github.com/isabelle-utp/utp-main/blob/90ec1d65d63e91a69fbfeeafe69bd7d67f753a47/tutorial/utp_csp_mini_mondex.thy\#L208} $$\rc{\dom(accts) = cs}{\begin{array}{c}\trace \neq \langle\rangle \land last(\trace) = pay.i.j.k \land \\ n
    \le accts(i) \implies accept.i \notin ref' \end{array}}{\truer}$$
\end{theorem}
\begin{proof}
  This property is different as it involves the pericondition rather than the postcondition. This is because we are
  reasoning about offered events in intermediate observations. We need to show in the pericondition that if the last
  event to occur is $pay.i.j.n$, and a sufficient amount is in account $i$ then we must not refuse to the accept the
  payment. This can be achieved by case analysis on pericondition $P_2$.
\end{proof}

\noindent Thus we have shown how our design contracts can be used to verify properties of simple \Circus actions. In the
next section we explore the UTP theory's healthiness conditions behind our contracts --- which provides support for this
verification.

%% file: sec/grd.tex
In this section, we present our UTP theory of reactive design contracts in detail. We describe the healthiness
conditions, core signature definitions, and algebraic laws, which substantiate those given in
Section~\ref{sec:contracts}. In particular, we define healthiness conditions for two UTP theories: $\healthy{SRD}$ in
Section~\ref{sec:srd}, which is a recasting of the previous reactive design theory from~\cite{Oliveira&09}, and
$\healthy{NSRD}$ in Section~\ref{sec:nsrd}, a novel healthiness condition that refines $\healthy{SRD}$ with additional
constraints to support reactive preconditions and invisible intermediate states. This latter healthiness condition is
the foundation of our reactive contract theory. We consider the algebraic law this theory supports, and the restrictions
it places on expressible operators. In Section~\ref{sec:recurse}, we consider the formalisation of recursion, and show
how Kleene's fixed-point theorem~\cite{Lassez82} can be applied to calculate tail recursive reactive designs. Finally in
Section~\ref{sec:parallel}, we detail our results in the formalisation of parallel composition.

\subsection{UTP Theory Preliminaries}

The theorems presented in Section~\ref{sec:contracts}, present reactive contracts using the syntactic form of
$\rc{P_1}{P_2}{P_3}$. It is often inconvenient to rely on such a specific syntactic form to reason about contracts, for
example when stating algebraic theorems. It is much more convenient to a reactive contract $P$ as a relation that admits
certain properties. Consequently, as usual in the UTP approach, we define healthiness conditions that characterise
well-formed contract predicates. This allows us to obtain a large number of algebraic laws from lattice theory and
related domains, and also allows us to reason about contracts without the need for the syntactic form.

In order to reason about contracts as opaque relations, we need a way of extracting the three parts from a contract. We
therefore define the following three functions:

\begin{definition}[Pre-, Peri-, and Postcondition Extraction Functions] \label{def:preperipost} \isalink{https://github.com/isabelle-utp/utp-main/blob/90ec1d65d63e91a69fbfeeafe69bd7d67f753a47/theories/rea_designs/utp_rdes_triples.thy\#L131}
  \begin{align*}
    \preR{P} &~~\defs~~ \negr P[true,false,false/ok,ok',wait] \\
    \periR{P} &~~\defs~~ P[true,true,false,true/ok,ok',wait,wait'] \\
    \postR{P} &~~\defs~~ P[true,true,false,false/ok,ok',wait,wait']
  \end{align*}
\end{definition}

\noindent These three functions variously substitute the observational variables to obtain the respective characteristic
predicates. In order to illustrate this, we first expand out the definition of reactive contract from
Definition~\ref{def:rcontract}, and also using Definitions~\ref{def:design} and \ref{def:reahealths}:
$$\rc{P_1}{P_2}{P_3} = \healthy{R1}(\healthy{R2}_c(\conditional{\IIsrd}{wait}{((ok \land P_1) \implies (ok' \land (\conditional{P_2}{wait'}{P_3})))}))$$
\noindent From this, we can see that the crucial variables that dictate the various behaviours are $wait$, $ok$, $ok'$,
and $wait'$. Variable $wait$ flags whether a sequential predecessor is in an intermediate state or has terminated, $ok$
flags whether a predecessor diverged, and $ok'$ and $wait'$ record this information for the present relation.

The precondition ($P_1$) is $\false$ when a reactive design was started ($ok \land \neg wait$) but it
diverged ($\neg ok'$). Thus, to extract the precondition we set $ok$, $ok'$, and $wait$ to $true$, $false$, and $false$,
respectively, and negate the result using reactive negation. The pericondition is reached when the reactive design was
started and did not diverge ($ok'$), but has not yet reached its final state ($\neg wait'$), and is thus
intermediate. The postcondition is similarly obtained, but $wait'$ becomes $false$. With these definitions we can prove
theorems that allow us to extract the constituent reactive relations from a syntactic reactive contract.

\begin{theorem}[Extracting Pre, Peri, and Postconditions] \label{thm:exttrip} \isalink{https://github.com/isabelle-utp/utp-main/blob/90ec1d65d63e91a69fbfeeafe69bd7d67f753a47/theories/rea_designs/utp_rdes_triples.thy\#L387}
  \begin{align*}
    \preR{\rc{P_1}{P_2}{P_3}} &~~=~~ P_1  \\
    \periR{\rc{P_1}{P_2}{P_3}} &~~=~~ P_1 \rimplies P_2 \\
    \postR{\rc{P_1}{P_2}{P_3}} &~~=~~ P_1 \rimplies P_3 
  \end{align*}
  Provided $P_1$, $Q_1$, and $Q_2$ are all $\healthy{RR}$ healthy.
\end{theorem}

\noindent The pericondition and postcondition are viewed through the prism of the precondition being satisfied, which is
why the implications are present in the result above. This is an important assumption of reactive designs. Although it
is possible to specify behaviours in the peri and postcondition outside the precondition, once the contract is
constructed these behaviours are pruned, so that the behaviour of a contract when its precondition is violated is always
$\Chaos$.

\subsection{Stateful Reactive Designs}
\label{sec:srd}

In this section, we introduce the first of our two reactive design theories: stateful reactive designs. Motivated by the
previous work on CSP and \Circus in UTP~\cite{Hoare&98,Cavalcanti&06,Oliveira&09}, we define the following healthiness
conditions for our generalised theory of reactive designs. We generalise the previous
works~\cite{Hoare&98,Cavalcanti&06,Oliveira&09} due to the underlying abstract trace algebra, and the extensible
alphabet of our theory~\cite{Foster17a}.

\begin{definition}[Stateful Reactive Designs Healthiness Conditions] \label{def:rdhcond} \isalink{https://github.com/isabelle-utp/utp-main/blob/90ec1d65d63e91a69fbfeeafe69bd7d67f753a47/theories/rea_designs/utp_rdes_healths.thy\#L401}
\begin{align*}
  \healthy{RD1}(P) &~~\defs~~ ok \rimplies P \\
  \healthy{RD2}(P) &~~\defs~~ P \relsemi \ckey{J} \\
  \healthy{SRD}(P) &~~\defs~~ \healthy{RD1} \circ \healthy{RD2} \circ \healthy{R}_s
\end{align*}
\end{definition}

\noindent $\healthy{RD1}$ and $\healthy{RD2}$ are analogous to $\healthy{H1}$ and $\healthy{H2}$ from the theory of
designs. Moreover, $\healthy{RD1}$ and $\healthy{RD2}$ correspond to $\healthy{CSP1}$ and $\healthy{CSP2}$ from the
theories of CSP~\cite{Hoare&98,Cavalcanti&06} and \Circus~\cite{Oliveira&09}. We rename them, firstly because our theory
has a different alphabet founded upon our trace algebra, and secondly because we do not specify the corresponding
$\healthy{CSP3}$ and $\healthy{CSP4}$, which constrain $ref$ and are thus specific to CSP and \Circus.

$\healthy{RD1}$ states, like $\healthy{H1}$, that observations are possible only after initiation, indicated by
$ok$. However, unlike $\healthy{H1}$, if $ok$ is false the resulting predicate is not $\true$, but $\truer$; that is, the
trace must monotonically increase, but the behaviour is otherwise unpredictable. $\healthy{RD2}$ is identical to
$\healthy{H2}$ and thus also $\healthy{CSP2}$.

Our overall healthiness condition for stateful reactive designs is then called $\healthy{SRD}$, which includes
$\healthy{RD1}$, $\healthy{RD2}$, and $\Rs$. Like $\healthy{RD1}$ and $\healthy{RD2}$, it is a generalisation of the
overall \Circus healthiness condition $\healthy{CSP}$~\cite{Oliveira&09}. We can next prove that $\healthy{SRD}$
relations admit a certain syntactic formulation using reactive contracts, as shown in the following theorem.

\begin{theorem}[Stateful Reactive Design Formulation] \label{thm:rdesform} \isalink{https://github.com/isabelle-utp/utp-main/blob/90ec1d65d63e91a69fbfeeafe69bd7d67f753a47/theories/rea_designs/utp_rdes_healths.thy\#L401}
  \begin{align*}
    \healthy{SRD}(P) &= \rc{\preR{P}}{\periR{P}}{\postR{P}}
  \end{align*}
\end{theorem}

\noindent If a relation is $\healthy{SRD}$ healthy, then it takes the form of a reactive contract with a pre-, peri-,
and postcondition. Theorem~\ref{thm:rdesform} thus allows us to take an $\healthy{SRD}$ predicate and deconstruct it
into its three parts, using the functions defined in Definition~\ref{def:preperipost}, which can then be manipulated
separately. This theorem shows an equivalence between the syntactic formulation, and elements of
$\theoryset{\healthy{SRD}}$: any element of the latter can be constructed using the contract triple
notation. Conversely, a well-formed reactive design triple always yields a healthy reactive design. This is confirmed by
the following closure theorem, which is a corollary of Theorems~\ref{thm:rdesform} and \ref{thm:exttrip}.

\begin{corollary} If $P_1$, $P_2$, and $P_3$ are all $\healthy{RR}$ healthy then $\rc{P_1}{P_2}{P_3}$ is $\healthy{SRD}$
  healthy. \label{thm:rcclosure} \isalink{https://github.com/isabelle-utp/utp-main/blob/90ec1d65d63e91a69fbfeeafe69bd7d67f753a47/theories/rea_designs/utp_rdes_triples.thy\#L568}
\end{corollary}

\noindent In order to show that a reactive contract is $\healthy{SRD}$, it suffices to show that its pre-, peri-, and
postcondition are all $\healthy{RR}$. A further corollary is the three triple extraction functions all produce healthy
reactive relations when applied to an $\healthy{SRD}$ relation.

\begin{corollary} If $P \is \healthy{SRD}$ then $\preR{P}$, $\periR{P}$, and $\postR{P}$ are all $\healthy{RR}$-healthy. \isalink{https://github.com/isabelle-utp/utp-main/blob/90ec1d65d63e91a69fbfeeafe69bd7d67f753a47/theories/rea_designs/utp_rdes_triples.thy\#L346}
\end{corollary}

Having defined a candidate theory for our reactive contracts, we will now explore its algebraic properties. The class of
$\healthy{SRD}$ relations admits a number of useful identities that we outline below.

\begin{theorem}[$\healthy{SRD}$ Laws] If $P$ is $\healthy{SRD}$ healthy then \isalink{https://github.com/isabelle-utp/utp-main/blob/90ec1d65d63e91a69fbfeeafe69bd7d67f753a47/theories/rea_designs/utp_rdes_triples.thy\#L1067}
  \begin{align}
    \IIsrd \relsemi P ~~=~~& P \\
    \Chaos \relsemi P ~~=~~& \Chaos \\
    \Miracle \relsemi P ~~=~~& \Miracle
  \end{align}
\end{theorem}

\noindent $\IIsrd$ is a left unit of any $\healthy{SRD}$ relation, and $\Chaos$ and $\Miracle$ are both left
annihilators. However, in general $\IIsrd$ is not a right unit, which the following theorem illustrates.

\begin{theorem} \label{thm:srdskip} If $P_1$, $P_2$, and $P_3$ are $\healthy{RR}$ then
  $\rc{\!P_1\!}{\!P_2\!}{\!P_3\!} \relsemi \IIsrd = \rc{\negr (\negr P_1) \relsemi \truer}{\exists \state' @ P_2}{P_3}$ \isalink{https://github.com/isabelle-utp/utp-main/blob/90ec1d65d63e91a69fbfeeafe69bd7d67f753a47/theories/rea_designs/utp_rdes_triples.thy\#L1079}
\end{theorem}

\noindent Right composition with $\IIsrd$ effectively imposes two additional requirements on the reactive contract: (1)
that the negated precondition does not refer to dashed variables other than $tr'$, which must be extension closed --- it
is $\healthy{RC1}$ (cf. Definition~\ref{def:RC}); and (2) that the pericondition does not refer to $\state'$, as
characterised by the existential quantifier. The latter follows as a direct consequence of $\healthy{R3}_h$~\cite{BGW09}
(cf. Definition~\ref{def:reahealths}), which requires that $\state$ must not be restricted in an intermediate
state. Thus, composition of $P$ with $\IIsrd$ has the same effect on $\state'$ in $P$ since $\IIsrd$ ignores any
intermediate states ($wait' = true$), rather than passing them forward to a successor process.

The former restriction is a direct consequence of healthiness condition $\healthy{RD1}$. To explain why, we consider the
reactive design $P \relsemi \IIsrd$ in the situation when $P$ has started ($ok \land wait$), but diverged ($\neg
ok'$). The following calculation shows the result.

\begin{example}[Divergence and $\healthy{RD1}$]
\begin{align*}
  &\quad~ (P \land ok \land \neg ok' \land \neg wait) \relsemi \IIsrd \\[.3ex]
  &= P[true, false, false/ok, ok', wait'] \relsemi \IIsrd[false/ok] & \text{[relational calculus]}  \\[.3ex]
  &= (\negr \preR{P}) \relsemi \IIsrd[false/ok] & \text{[Definition~\ref{def:preperipost}]} & \text{} \\[.3ex]
  &= (\negr \preR{P}) \relsemi (\healthy{RD1}(\IIsrd))[false/ok] & \text{[Theorem~\ref{thm:rcclosure}]} \\[.3ex]
  &= (\negr \preR{P}) \relsemi (ok \rimplies \IIsrd)[false/ok] & \text{[Definition~\ref{def:rdhcond}]} \\[.3ex]
  &= (\negr \preR{P}) \relsemi \truer & \text{[predicate calculus]}
\end{align*}
\end{example}

\noindent When $P$ diverges with $\neg ok'$, its behaviour is $\negr \preR{P}$: the precondition has been
violated. Moreover, if $ok'$ is false in $P$ then $ok$ is false in $\IIsrd$, which yields the relation
$\truer$. Actually, this follows for any $\healthy{RD1}$ relation, not just $\IIsrd$, since
$\healthy{RD1}[false/ok] = \truer$. Consequently, right composition with any $\healthy{SRD}$ relation yields the same
result: the precondition should be $\healthy{RC1}$.

Consequently, the right composition law motivates that we need to consider a refined theory for reactive contracts, in
order to support the desirable algebraic and calculational laws of Section~\ref{sec:contracts}.  This is the objective
of the following section.

\subsection{Normal Stateful Reactive Designs}
\label{sec:nsrd}

We refine $\healthy{SRD}$ by identifying the subclass of normal stateful reactive designs, $\healthy{NSRD}$, which
provides the theory domain of our reactive contracts.

\begin{definition}[Normal Stateful Reactive Designs Healthiness Conditions] \label{def:nsrddef} \isalink{https://github.com/isabelle-utp/utp-main/blob/90ec1d65d63e91a69fbfeeafe69bd7d67f753a47/theories/rea_designs/utp_rdes_normal.thy\#L11}
\begin{align*}
  \healthy{RD3}(P) &~~\defs~~ P \relsemi \IIsrd \\
  \healthy{NSRD}(P) &~~\defs~~ \healthy{RD1} \circ \healthy{RD3} \circ \healthy{R}_s
\end{align*}
\end{definition}

\noindent $\healthy{RD3}$ is analogous to $\healthy{H3}$: it requires that skip is a right unit. This ensures that the
precondition is a reactive condition and the pericondition does not depend on $\state'$. The use of the word ``normal''
here is therefore by analogy with normal designs~\cite{Guttman2010}. $\healthy{NSRD}$ does not explicitly invoke
$\healthy{RD2}$ as it is subsumed by $\healthy{RD3}$, as the following theorem demonstrates.

\begin{theorem}[$\healthy{RD3}$ subsumes $\healthy{RD2}$] \label{thm:rd3subsumes} $\healthy{RD2}(\healthy{RD3}(P)) = \healthy{RD3}(\healthy{RD2}(P)) = \healthy{RD3}(P)$ \isalink{https://github.com/isabelle-utp/utp-main/blob/90ec1d65d63e91a69fbfeeafe69bd7d67f753a47/theories/rea_designs/utp_rdes_normal.thy\#L31}
\end{theorem}

\begin{proof}
  This follows since $\ckey{J} \relsemi \IIsrd = \IIsrd \relsemi \ckey{J} = \IIsrd$.
\end{proof}

\noindent Consequently, it is easy to show that every $\healthy{NSRD}$-healthy relation is also $\healthy{SRD}$-healthy,
and therefore the theorems of Section~\ref{sec:srd} remain valid. We can also prove the following theorem that follows
as a consequence of Theorems~\ref{thm:rd3subsumes} and \ref{thm:srdskip}.

\begin{theorem}[Normal Stateful Reactive Design Formulation] \label{thm:nrdesform} \isalink{https://github.com/isabelle-utp/utp-main/blob/90ec1d65d63e91a69fbfeeafe69bd7d67f753a47/theories/rea_designs/utp_rdes_normal.thy\#L153}
  \begin{align*}
    \healthy{NSRD}(P) &= \rc{\healthy{RC1}(\preR{P})}{(\exists \state' @ \periR{P})}{\postR{P}}
  \end{align*}
\end{theorem}

In addition to ensuring that the elements of the triple of $\healthy{RR}$ healthy, $\healthy{NSRD}$ also ensures that
the precondition is a reactive condition, and that the pericondition does not refer to $\state'$. 

\begin{corollary} \label{thm:nsrdintro} $P$ is $\healthy{NSRD}$ healthy provided that the following
  conditions hold: \isalink{https://github.com/isabelle-utp/utp-main/blob/90ec1d65d63e91a69fbfeeafe69bd7d67f753a47/theories/rea_designs/utp_rdes_normal.thy\#L304}
  \begin{enumerate}
    \item $P$ is $\healthy{SRD}$ healthy;
    \item $\preR{P}$ is $\healthy{RC}$ healthy;
    \item $\periR{P}$ does not mention $\state'$.
  \end{enumerate}
\end{corollary}

\noindent A further corollary finally justifies the syntactic formulation given in Section~\ref{sec:contracts}:

\begin{corollary}[$\healthy{NSRD}$ contract closure] $ $ \label{thm:nrcclosure} \isalink{https://github.com/isabelle-utp/utp-main/blob/90ec1d65d63e91a69fbfeeafe69bd7d67f753a47/theories/rea_designs/utp_rdes_normal.thy\#L310}

\noindent If $P_1$ is $\healthy{RC}$, $P_2$ and $P_3$ are both $\healthy{RR}$, and $P_2$ does not refer to $\state'$, then $\rc{P_1}{P_2}{P_3}$ is $\healthy{NSRD}$ healthy. 
\end{corollary}

\noindent We now, therefore, have an adequate formulation of the reactive design theory. In spite of its restrictions,
$\healthy{SRD}$ remains useful as a means of obtaining algebraic theorems. In the remainder of this section we will
therefore explore the algebraic properties of the two theories.

\subsection{Algebraic Properties}
\label{sec:grd-alg}

$\healthy{SRD}$ and $\healthy{NSRD}$ are both idempotent and continuous, and therefore both form complete lattices, as
stated in the following theorem.

\begin{theorem}[Reactive Design Lattices] $\healthy{SRD}$ and $\healthy{NSRD}$ healthy predicates form complete lattices
  with $\top_\srdes \defs \healthy{SRD}(\false) = \Miracle$ and $\bot_\srdes \defs \healthy{SRD}(\true) =
  \Chaos$. \label{thm:rdes-lattice} \isalink{https://github.com/isabelle-utp/utp-main/blob/90ec1d65d63e91a69fbfeeafe69bd7d67f753a47/theories/rea_designs/utp_rdes_normal.thy\#L732}
\end{theorem}
\begin{proof} Standard proofs of idempotency for $\healthy{CSP1}$ and $\healthy{CSP2}$~\cite{Hoare&98,Cavalcanti&06}
  apply also to $\healthy{RD1}$ and $\healthy{RD2}$, respectively. $\healthy{RD3}$ is idempotent since
  $\IIsrd \relsemi \IIsrd = \IIsrd$. $\healthy{RD1}$ is continuous since it is disjunctive. $\healthy{RD2}$ and
  $\healthy{RD3}$ are both continuous since sequential composition distributes through infima to the left and right.
\end{proof}

\noindent We note that both $\Miracle$ and $\Chaos$, following the form given in Theorem~\ref{thm:nrcclosure}, are both
$\healthy{NSRD}$ healthy. Through the Knaster-Tarski theorem~\cite{Tarski55}, we also obtain weakest fixed-point
operators $\thlfp{R}$ and $\thlfp{N}$ for the two theories, which we will further explore in
Section~\ref{sec:recurse}. Since $\healthy{SRD}$ and $\healthy{NSRD}$ are also continuous, by Theorem~\ref{thm:contthm}
we can rewrite the weakest fixed-points to $\mu X @ F(\healthy{SRD}(X))$ and $\mu X @ F(\healthy{NSRD}(X))$,
respectively. Thus, we can reason about recursive reactive designs using the relational calculus lattice rather than the
theory specific ones. We can also show that reactive designs are closed under the standard relational calculus
operators, as the following theorems demonstrate.

\begin{theorem}[Reactive Design Composition] \label{sec:rdes-seq} If $P$ and $Q$ are $\healthy{NSRD}$ healthy then \isalink{https://github.com/isabelle-utp/utp-main/blob/90ec1d65d63e91a69fbfeeafe69bd7d67f753a47/theories/rea_designs/utp_rdes_normal.thy\#L346}
  \begin{align*}
  P \relsemi Q = \Rs 
    \left(\begin{array}{l}
            \preR{P} \land (\postR{P} \wpR \preR{Q}) \\[.5ex]
            \shows~~ \periR{P} \lor (\postR{P} \relsemi \periR{Q}) \\[.5ex]
            \wcond~~ \postR{P} \relsemi \postR{Q}
          \end{array}\right)
  \end{align*}
\end{theorem}

\noindent Theorem~\ref{sec:rdes-seq} is essentially the same as Theorem~\ref{thm:rdes-seq}, but relies on healthiness of
$P$ and $Q$, rather than their syntactic form. As similar law holds for $\healthy{SRD}$-healthy predicates, though with
a more complex precondition and pericondition, following the form given in Theorem~\ref{thm:srdskip}. We finally show
closure of the theory under the main programming operators.

\begin{theorem}[Reactive Designs Closure]
  $\healthy{SRD}$ and $\healthy{NSRD}$ healthy predicates are closed under $\intchoice$, $\sqcup$, $\conditional{}{b}{}$,
  ~$\relsemi$~, $\IIsrd$, $\assignsR{\sigma}$, $\Miracle$, and $\Chaos$. \isalink{https://github.com/isabelle-utp/utp-main/blob/90ec1d65d63e91a69fbfeeafe69bd7d67f753a47/theories/rea_designs/utp_rdes_normal.thy\#L386}
\end{theorem}

\noindent This closure theorem means that we can import the algebraic laws for the relational calculus operators from
core UTP~\cite{Hoare&98}, such as several equations of Theorem~\ref{thm:rellaws}. Finally, we note that our theory
admits the following familiar laws from relational calculus for assignment.

\begin{theorem}[Reactive Design Assignment Compositions] \label{thm:rdes-asn-comp} If $P$ is $\healthy{NSRD}$ healthy,
  $x$ is a state variable, and $v$ is an expression containing only state variables, then: \isalink{https://github.com/isabelle-utp/utp-main/blob/90ec1d65d63e91a69fbfeeafe69bd7d67f753a47/theories/rea_designs/utp_rdes_prog.thy\#L990}
  \begin{align}
    (\assignR{x}{v})    ~ \relsemi ~ P ~~=~~& P[v/x] \label{thm:rdes-asn-comp-1} \\[.5ex]
    \assignsR{\sigma} ~ \relsemi ~ P ~~=~~& \substapp{\sigma}{P} \label{thm:rdes-asn-comp-2}
  \end{align}
\end{theorem}

\noindent Theorem~\ref{thm:rdes-asn-comp-1}, the more usual law~\cite{Hoare87}, shows that a sequential assignment can
be turned into a substitution applied to its successor. It is an instance of the more general
Theorem~\ref{thm:rdes-asn-comp-2} when $\sigma = \{x \mapsto v\}$, crucially, provided that $x$ is a state variable and
not an arbitrary UTP variable. 

Theorem~\ref{thm:rdes-asn-comp} depends directly on the use of $\healthy{R3}_h$, and the resulting requirement that
periconditions do not refer to after state. It would not hold if allowed such references, for example by substituting
$\healthy{R3}_h$ with the original healthiness condition $\healthy{R3}$~\cite{Hoare&98}. In order to understand why,
consider the following corollary of Theorem~\ref{thm:rdes-asn-comp-1} involving a \Circus event prefix:

\begin{theorem}[Assignment and Events] \label{thm:asnevent}
   $$(\assignR{x}{e} \relsemi c \then P) ~~=~~ (c[e/x] \then (\assignR{x}{e} \relsemi P))$$
\end{theorem}

\noindent Theorem~\ref{thm:asnevent} allows us to push assignment through event prefixes, whilst making an appropriate
substitution.  It can be found in reactive languages like Occam~\cite{Roscoe1988}. This law does not hold in previous
\Circus semantics~\cite{Oliveira&09} as $\state'$ (there called $v'$) is revealed in intermediate states. Thus, whilst
the left hand side of this equation admits an intermediate observation where $x$ is updated to $e$, because the
pericondition does characterise state updates, the right hand side does not.
 
There are, however, costs to this simplification. A side effect is to prevent encoding of McEwan's \Circus interruption
operator~\cite{McEwan06}, which has the form $P \mathop{\triangle_i} Q$. This operator is similar to
$P \interleave (i \then Q)$, except that if the interruption event $i$ occurs before $P$ terminates, and whilst it is
quiescent, then the remaining behaviour of $P$ is pruned. If $P$ terminates, then conversely the behaviours of $Q$ are
lost. Crucially, the internal state of $P$ is retained following interruption~\cite{BGW09}, and is passed on to $Q$,
which is unique to McEwan's work (cf.~\cite{Sherif2010,Wei2013,Woodcock14}). For example, consider the following action:
$$(x := 5 \relsemi a \then \Skip) \mathop{\triangle_b}\, (x := x + x).$$
First, we note that the leading assignment $x := 5$ cannot be interrupted, since it does not have a quiescent
state. However, once this has occurred the left-hand side enters a quiescent state where $a$ is enabled. Through the
interruption operator, $b$ is also enabled at this point. If $a$ occurs then the entire action terminates in a final
state with $x = 5$. However, if $b$ occurs then the value of $x$ at this point is retained, and the action $x := x + x$
occurs, leading to final value of $10$ for $x$. 

We see, therefore, how the value of $x$ is retained and used by the interruption action. Such a retention of the state,
however, cannot be represented in the presence of $\healthy{RD3}$ as no intermediate state variables are
recorded. Therefore, if such an interruption operator is required, the loss of Theorem~\ref{thm:asnevent} must be
accepted, and $\healthy{R3}$ used as the base of reactive designs instead. This is a design choice depending on the kind
of reasoning and expressivity needed, as $\healthy{R3}$ reactive designs carry more information and so are more
distinguishing than $\healthy{R3}_h$ reactive designs. $\healthy{R3}$ is supported in Isabelle/UTP if its use is
desired.

\subsection{Recursion}
\label{sec:recurse}
\input{sec/recurse.tex}

\subsection{Parallel Composition}
\label{sec:parallel}

In this section we introduce the parametric parallel composition operator for reactive designs, and substantiate a
number of preliminary results for the operator, including well-formedness conditions supported by a novel healthiness
condition. Parallel composition in UTP is expressed in terms of the parallel-by-merge scheme $P \mathop{\parallel_M} Q$,
whereby the final states of concurrent separated processes $P$ and $Q$ are merged by predicate $M$. We adopt a slightly
simplified definition of parallel-by-merge, originally presented in~\cite{Foster16a}, which assumes that the alphabet of
both $P$ and $Q$ is the same, but is otherwise semantically equivalent to the standard UTP definition~\cite[chapter
7]{Hoare&98}.

\begin{definition}[Parallel-by-Merge] \label{def:par-by-merge} \isalink{https://github.com/isabelle-utp/utp-main/blob/07cb0c256a90bc347289b5f5d202781b536fc640/utp/utp_concurrency.thy\#L275}
  $$P \parallel_M Q ~\defs~ (\psep{P}{0} \land \psep{Q}{1} \land v' = v) \relsemi M$$
\end{definition}

\noindent The intuition for parallel-by-merge is given in Figure~\ref{fig:pbm}. It splits the state space into three
parts, one of which is passed to $P$, the second to $Q$, and the third is simply copied from the original input
state. Assuming both $P$ and $Q$ produce an output, the merge predicate $M$ computes the overall output that should be
given by merging the three states. For imperative programs, a simple example merge predicate may pick two disjoint
subsets of the variables output by $P$ and $Q$ to produce an overall output. The value of such a scheme for parallel
composition is that different concurrency schemes can be supported, and properties of the parallel composition operator
often reduce to properties of the merge predicate~\cite[theorem 7.2.10, page 173]{Hoare&98}.

To formalise this in Definition~\ref{def:par-by-merge}, we require that $P$ and $Q$ are homogeneous relations with the
alphabet $\{v, v'\}$, where $v$ is a vector of variables. $\psep{P}{0}$ and $\psep{Q}{1}$ rename the dashed variables by
adding indices $0$ and $1$, respectively, so they can be distinguished by the merge predicate $M$\footnote{These are
  called ``separating simulations'' in~\cite[page 172]{Hoare&98}, and are denoted using special relations called $U0$
  and $U1$.}. $M$ is a heterogeneous alphabetised relation whose alphabet is $\{ 0.v, 1.v, v, v'\}$. It takes three
copies of the variables: the renamed dashed variables of $P$ and $Q$, and all undashed variables. The separated
processes are conjoined, which is acceptable because of their disjoint output alphabets, along with a predicate that
copies all initial variables ($v' = v$). The resulting conjunctive relation is then composed in sequence with $M$. This
merge predicate calculates the overall final state in terms of the initial state, and final states of $P$ and $Q$. Thus,
the overall composition $P \parallel_M Q$ has alphabet $\{v, v'\}$.

For our theory of reactive designs, the objective, as indicated by Section~\ref{sec:parcontract} is to specialise
parallel-by-merge so that it acts only on the trace, state, and other semantic observational variables (like
$\refu$). We therefore define the parametric merge predicate $\mathcal{M}_R(M)$ that merges $ok$ and $wait$ variables,
so handling divergence and intermediate observations. It defers merging of the states and traces to an ``inner merge
predicate'' $M$, such as $M_c$ described in Example~\ref{ex:interleave}. We require that $M$ does not refer to $ok$ or
$wait$ nor decorations thereof.

We define three auxiliary merge operators that construct the ``outer merge predicate'', by showing how $ok$ and $wait$
are merged, and imposition the reactive design healthiness conditions. These operators are then used to define the
reactive design parallel composition operator.

\begin{definition}[Reactive Designs Parallel Composition] \label{def:rdpar} \isalink{https://github.com/isabelle-utp/utp-main/blob/90ec1d65d63e91a69fbfeeafe69bd7d67f753a47/theories/rea_designs/utp_rdes_parallel.thy\#L107}
  \begin{align*}
    \mathcal{N}_0(M) &~\defs~ (wait' = (0.wait \lor 1.wait) \land tr \le tr' \land M) \\[1ex]
    \mathcal{N}_1(M) &~\defs~ (ok' = (0.ok \land 1.ok) \land \mathcal{N}_0(M)) \\[1ex]
    \mathcal{M}_R(M) &~\defs~ \healthy{RD3}(\healthy{RD1}(\healthy{R3}_h(\mathcal{N}_1(M)))) \\[1ex]
    P \rcpar{M} Q    &~\defs~ P \parallel_{\mathcal{M}_R(M)} Q
  \end{align*}
\end{definition}

\noindent The auxiliary merge functions $\mathcal{N}_0$ and $\mathcal{N}_1$ conjoin the inner merge predicate $M$ with
three conjuncts. $\mathcal{N}_0$ handles merging of the $wait$ variables. If either $P$ or $Q$ admits an intermediate
observation ($0.wait$ or $1.wait$), then also the composite observation is intermediate, and thus we take the
disjunction of the $wait$ variables to determine $wait'$. The third conjunct ensures the resulting merge is
$\healthy{R1}$ healthy. $\mathcal{N}_1$ handles merging of the $ok$ variables. If either $P$ or $Q$ diverges, then we
require that their composition also diverges and thus $ok'$ takes the conjunction of both $ok$ variables. Using
$\mathcal{N}_1$, we then define the overall merge predicate $\mathcal{M}_R$ by application of three healthiness
conditions to construct a stateful reactive design. The parametric merge of two reactive designs, $P \rcpar{M} Q$ with
inner merge predicate $M$, is simply a parallel-by-merge using $\mathcal{M}_R(M)$.

It remains to prove that $\healthy{NSRD}$ is closed under reactive design parallel composition. In order to prove this,
we need to restrict the form of the merge predicate using new healthiness conditions. Firstly, we need a modified
version of $\healthy{R2}_c$ that is applicable to merge predicates, first defined in~\cite{Foster17a}.

\begin{definition}[$\healthy{R2}_c$ for Merge Predicates] \isalink{https://github.com/isabelle-utp/utp-main/blob/90ec1d65d63e91a69fbfeeafe69bd7d67f753a47/theories/reactive/utp_rea_parallel.thy\#L21}
  \begin{align*}
    \healthy{R2}_m(M) &\defs (P[\tempty,tr'\!-\!tr,0.tr\!-\!tr,1.tr\!-\!tr/tr,tr',0.tr,1.tr]) \infixIf tr \le tr' \infixElse P
  \end{align*}
\end{definition}

\noindent Merge predicates have three ways of accessing the trace history through the three respective copies of the
trace variable. Thus, it is necessary to delete the history in $0.tr$, $1.tr$, and $tr$ in the healthiness condition
$\healthy{R2}_m$ to ensure that this does not occur. This allows us to prove the following theorem

\begin{theorem} \label{thm:parreaclos}
  $P \parallel_M Q$ is $\healthy{R1}$ and $\healthy{R2}_c$ healthy provided that $P$ and $Q$ are both $\healthy{R1}$ and
  $\healthy{R2}_c$, and $M$ is $\healthy{R1}$ and
  $\healthy{R2}_m$. \isalink{https://github.com/isabelle-utp/utp-main/blob/90ec1d65d63e91a69fbfeeafe69bd7d67f753a47/theories/reactive/utp_rea_parallel.thy\#L114}
\end{theorem}

This theorem provides the circumstances under which a parallel-by-merge constructs a healthy reactive process. From
this, we define the following healthiness condition for reactive-design inner merge predicates.

\begin{definition}[Healthy Inner Merge Predicate] \isalink{https://github.com/isabelle-utp/utp-main/blob/90ec1d65d63e91a69fbfeeafe69bd7d67f753a47/theories/rea_designs/utp_rdes_parallel.thy\#L128}
  $$\healthy{RDM}(M) \defs \healthy{R2}_m(\exists 0.ok, 1.ok, ok, ok', 0.wait, 1.wait, wait, wait' @ M)$$
\end{definition}

\noindent $\healthy{RDM}$ contains $\healthy{R2}_m$ and additionally formalises the requirement that the inner merge
predicate does not contain $ok$, $wait$, and decorations thereof, through existential quantification. These variables
have already been handled by the outer merge predicate, and so the inner merge predicate should not refer to them. $M_c$
in Example~\ref{ex:interleave} is clearly $\healthy{RDM}$, for example, since it refers only to $0.\trace$ and
$1.\trace$, thus satisfying $\healthy{R2}_m$, and does not mention $ok$ or $wait$ in any way.

Using this definition of healthy reactive inner merges, we can then prove the following closure theorem for parallel
composition:

\begin{theorem}[Parallel Composition Closure] \label{thm:parclos} If $P$ and $Q$ are $\healthy{SRD}$ healthy, and $M$ is
  $\healthy{RDM}$ healthy, then $P \parallel_M Q$ is $\healthy{NSRD}$ healthy. \isalink{https://github.com/isabelle-utp/utp-main/blob/90ec1d65d63e91a69fbfeeafe69bd7d67f753a47/theories/rea_designs/utp_rdes_parallel.thy\#L153}
\end{theorem}

\noindent As seen in Definition~\ref{def:rdpar}, parallel composition, like sequential composition, is not defined
explicitly as a reactive design using the contract syntax. However, in order to calculate the meaning of a parallel
reactive program it is necessary to know how to calculate a contract form for it. By Theorem~\ref{thm:parclos} we know
that parallel composition is a normal stateful reactive design ($\healthy{NSRD}$), and therefore we can invoke
Theorem~\ref{thm:nrdesform} to split it into a reactive design triple, toward substantiation of
Theorem~\ref{thm:rdespar}. It then suffices to calculate its pre-, peri-, and postcondition. In order to do this, we
need to characterise interference between parallel processes using a notion of weakest rely condition. This is the
weakest context under which two parallel composed processes do not violate one another's assumptions.

\begin{definition}[Weakest Rely Condition] \label{def:wppR} \isalink{https://github.com/isabelle-utp/utp-main/blob/90ec1d65d63e91a69fbfeeafe69bd7d67f753a47/theories/rea_designs/utp_rdes_parallel.thy\#L667}
  \begin{align*}
    P \wppR{M} Q ~~\defs~~& \negr ((\negr Q) \parallel_{\hbox{\small $M$} \relsemi \truer} P)
  \end{align*}
\end{definition}

\noindent The operator $P \wppR{M} Q$ is essentially the concurrent case of the reactive weakest precondition $\wpR$
(cf. Definition~\ref{def:wpR}). It represents the weakest context where reactive relation $P$ does not lead to the
violation of reactive condition $Q$. We achieve this by first merging the possible traces of negated $Q$ with those of
$P$ using the merge predicate $M$ composed with $\truer$. This determines all the behaviours permitted by merge
predicate $M$ that are enabled by $P$ and yet violate $Q$. The composition with $\truer$ ensures that resulting reactive
relation is extension closed. We then negate the resulting relation to obtain the overall reactive precondition.

We can show that the weakest rely condition constructor forms a reactive condition, which is essential to our using it
in the reactive contract precondition:

\begin{theorem}[Weakest Rely Condition forms a Reactive Condition] If $P$ and $Q$ are $\healthy{RR}$ healthy, and $M$ is
  $\healthy{RDM}$ healthy, then $P \wppR{M} Q$ is a reactive condition. \label{thm:wrcrc} \isalink{https://github.com/isabelle-utp/utp-main/blob/90ec1d65d63e91a69fbfeeafe69bd7d67f753a47/theories/rea_designs/utp_rdes_parallel.thy\#L711}
\end{theorem}

\begin{proof}
  By Theorem~\ref{thm:parreaclos} we can show that $P \wppR{M} Q$ is $\healthy{R1}$ and $\healthy{R2}_c$. It
  therefore suffices to show that it is also $\healthy{RC1}$ healthy, which we do below.
  \begin{align*}
    \healthy{RC1}(P \wppR{M} Q) 
    =~& \healthy{RC1}(\negr ((\negr P) \parallel_{\hbox{\small M} \relsemi \truer} Q)) & \text{[$\wppR{M}$ definition]} \\
    =~& \negr ((\negr \negr ((\negr P) \parallel_{\hbox{\small M} \relsemi \truer} Q)) \relsemi \truer) & \text{[$\healthy{RC1}$ definition]} \\
    =~& \negr (((\negr P) \parallel_{\hbox{\small M} \relsemi \truer} Q) \relsemi \truer) & \text{[predicate calculus]} \\
    =~& \negr (((\negr P) \parallel_{\hbox{\small M}} Q) \relsemi \truer \relsemi \truer) & \text{[$\parallel_{\hbox{\small M}}$ definition]} \\
    =~& \negr (((\negr P) \parallel_{\hbox{\small M}} Q) \relsemi \truer) & \text{[relational calculus]} \\
    =~& \negr ((\negr P) \parallel_{\hbox{\small M} \relsemi \truer} Q) & \text{[$\parallel_{\hbox{\small M}}$ definition]} \\
    =~& P \wppR{M} Q & \text{[$\wppR{M}$ definition]}
  \end{align*}
  This essentially follows due the imposition of extension closure in the merge predicate. \qedhere
\end{proof}

\noindent We can now use the weakest rely condition to calculate the precondition for parallel composition. We enumerate
all the ways that divergence could arise from the composition, and request that this not happen using weakest rely
conditions.

\begin{theorem}[Parallel Precondition] \label{thm:parpre} \isalink{https://github.com/isabelle-utp/utp-main/blob/90ec1d65d63e91a69fbfeeafe69bd7d67f753a47/theories/rea_designs/utp_rdes_parallel.thy\#L744}
  $$\preR{P \rcpar{M} Q} = 
    \left(\begin{array}{l} 
      \preR{P} \wppR{M} \periR{Q} \land \preR{P} \wppR{M} \postR{Q} \land \\[.5ex]
      \preR{Q} \wppR{M} \periR{P} \land \preR{Q} \wppR{M} \postR{P}
    \end{array}\right)$$
\end{theorem}

\noindent Divergence can arise in four possible possible ways: when an intermediate or final observation of $Q$ leads to
a state where the precondition of $P$ is violated, and the converse situation for $P$ and $Q$.

The theorem for calculating the pericondition of parallel composition requires that we merge the traces, but not the
state variables as these are concealed in intermediate observations. So, we define the following derived parallel
composition operator.

\begin{definition}[Intermediate Merge] $P \rcmergee{M} Q ~~\defs~~ P \parallel_{\exists st' @ M} Q$
\end{definition}

\noindent $P \rcmergee{M} Q$ merges the traces of $P$ and $Q$, whilst hiding the merged post-state using an existential
quantifier. It is used in the following calculation of parallel pericondition.

\begin{theorem}[Parallel Pericondition and Postcondition] \label{thm:parperipost} \isalink{https://github.com/isabelle-utp/utp-main/blob/90ec1d65d63e91a69fbfeeafe69bd7d67f753a47/theories/rea_designs/utp_rdes_parallel.thy\#L437}
  \begin{align*}
    \periR{P \rcpar{M} Q} ~=~~& \preR{P \rcpar{M} Q} \rimplies
    \left(\begin{array}{l} 
       \periR{P} \rcmergee{M} \periR{Q} \lor \\ 
       \postR{P} \rcmergee{M} \periR{Q} \lor \\
       \periR{P} \rcmergee{M} \postR{Q}
    \end{array}\right) \\[1ex]
    \postR{P \rcpar{M} Q} ~=~~& \preR{P \rcpar{M} Q} \rimplies (\postR{P} \parallel_M \postR{Q})
  \end{align*}
\end{theorem}

\noindent The merge predicate calculates the value of $wait'$ from the disjunction of $0.wait$ and $1.wait$, and so the
overall observation of a parallel composition is intermediate if either $P$ or $Q$ is. The parallel postcondition, for
the same reason, requires that both $P$ and $Q$ have reached their final states. Unlike for the pericondition, the
normal parallel by merge operator is used, since the state is no longer concealed.

We have now shown the laws for calculating the pre-, peri-, and postconditions of parallel reactive contracts. Combining
Theorems~\ref{thm:parpre} and \ref{thm:parperipost} with Theorem~\ref{thm:exttrip} allows us to finally substantiate
Theorem~\ref{thm:rdespar}. This then completes our preliminary results on parallel composition for reactive designs.

\vspace{1ex}

In this section, we have laid the foundations for our theory of reactive contracts, considering its foundations, links
to other theories, and recursion and parallel composition. We next consider its mechanisation in Isabelle/HOL, which
allowed us to mechanically prove the all of the laws previously presented.

%% file: sec/recurse.tex
In this section we will show how to calculate reactive contracts for a restricted class of recursive models, with the
particular aim of substantiating Theorem~\ref{thm:rdes-rec}. In Section~\ref{sec:grd-alg}, we have shown that
generalised reactive designs form a complete lattice. Thus, for any monotonic process constructor $F$ we can be sure
there exists fixed-points $\thlfp{R}\,F$ and $\thgfp{R}\,F$. However, in order to reason about recursive reactive
contracts generally, we need to calculate the pre, peri, and postconditions of such constructions.

In general, we are most interested in the weakest fixed-point for reactive designs, $\thlfp{R}\,F$, as the strongest
fixed-point yields miraculous behaviour for erroneous constructions~\cite{Cavalcanti04}. For example,
$(\thgfp{R}\, X @ X) = \Miracle$, whereas in reality an infinite loop is a programmer error that should yield $\Chaos$, which
$\thlfp{R}\,X @ X$ does. 

In order to calculate the reactive design of a weakest fixed-point we employ two results: (1) Hoare and He's proof that
guarded processes yield unique fixed-points~\cite[theorem 8.1.13, page 206]{Hoare&98}, and (2) Kleene's fixed-point
theorem~\cite{Lassez82}. The latter allows us to convert from a recursive construction with a strongest fixed-point to
an iterative construction, using a replicated internal choice of power constructions. Since we can calculate the
reactive design of replicated processes, we can therefore tackle recursion.

Hoare and He's theorem states, informally, that guarded functions on reactive processes have a unique fixed-point, that
is, if, for any relation $X$, $F(X)$ is guarded, then $\mu F = \nu F$. Guardedness is defined as follows.

\begin{definition}[Guarded Reactive Designs] \label{def:guarded} A function on reactive designs $F : \theoryset{\healthy{SRD}} \to \theoryset{\healthy{SRD}}$ is guarded provided that, for any
  $P \in \theoryset{\healthy{SRD}}$ and $n \in \mathbb{N}$, \isalink{https://github.com/isabelle-utp/utp-main/blob/90ec1d65d63e91a69fbfeeafe69bd7d67f753a47/theories/rea_designs/utp_rdes_guarded.thy\#L65}
  $$(F(P) \land gv(n+1)) = (F(P \land gv(n)) \land gv(n+1))$$ where
  $gv(n) \defs (tr \le tr' \land \#\trace < n)$ and $\# : \tset \to \mathbb{N}$ is a discrete trace measure function. We
  extend the trace algebra $(\tset, \tcat, \tempty)$ with the following axioms for the measure function; for
  $s, t \in \tset$:
  $$\#\tempty = 0 \qquad s > 0 \implies \# s > \tempty \qquad \#(s \tcat t) = \#s + \#t$$
\end{definition}

\noindent Definition~\ref{def:guarded} is similar to the one given in~\cite{Hoare&98}, but is generalised to allow a
variety of different discrete measure functions that satisfy the measure function axioms. An example measure function
that satisfies these axioms is the length function on sequences.

Given a reactive design $\mu X @ F(X)$, if $F$ is guarded according to Definition~\ref{def:guarded} then, intuitively,
before recursion variable $X$ can be reached, $F$ must have produced a non-empty portion of the trace. For example, in
CSP we may have a process $\mu X @ a \then X$, which is guarded since it must perform an $a$ before recursing. This is
ensured by requiring that the trace contribution of the function applied to a reactive design $F(P)$ yields a trace
strictly longer than that produced by $P$. This is the purpose of $gv$: if we observe $F(P)$ in a context where the
trace is longer than $n+1$ (enforced by $gv(n+1)$) then we can conclude that the trace contributed by $P$ must be no
longer than $n$, and thus we can conjoin it with $gv(n)$. We can use this to prove Hoare and He's theorem.

\begin{theorem}[Unique Fixed-Points] \label{thm:uniqfp}
  If $F$ is guarded then $\mu F = \nu F$. \isalink{https://github.com/isabelle-utp/utp-main/blob/90ec1d65d63e91a69fbfeeafe69bd7d67f753a47/theories/rea_designs/utp_rdes_guarded.thy\#L73}
\end{theorem}

\noindent Technically, the $\mu$ and $\nu$ operators we use here are those of the relational calculus lattice, and not
of an arbitrary UTP theory (hence the lack of subscripts). The proof given in~\cite[theorem 8.1.13, page 206]{Hoare&98},
which our mechanised proof follows, omits the step that transitions the theory fixed-point operator ($\thlfp{R}$) to the
relational one ($\mu$). However, this step is necessary in order to employ their approximation chain
theorem~\cite[theorem 2.7.6, page 63]{Hoare&98}. Thus, in order to employ Theorem~\ref{thm:uniqfp} for reactive designs,
we first have to use Theorems~\ref{thm:rdes-lattice} and \ref{thm:contthm} to convert the reactive design fixed-point
operator. Continuity of $\healthy{SRD}$ is thus an important property. Using that, for any guarded process, we can
convert a recursive construction using a weakest fixed-point to one using a strongest fixed-point.

In order to make use of Theorem~\ref{thm:uniqfp}, it is necessary to prove guardedness theorems for the operators of the
target language. In general, this can be quite complicated; in many cases, however, we can shortcut guardedness and
instead focus on tail-recursive fixed-point constructions of the form $\thlfp{R}\, X @ P \relsemi X$, where $X$ is not
mentioned in $P$, as employed by Theorem~\ref{thm:rdes-rec}. This pattern, though restrictive, covers a large number of
specifiable \Circus processes, for instance. In this case, guardedness can be shown simply by showing that $P$ always
produces events before it terminates. Of course, $P$ may not terminate at all, but in this case the recursion variable
$X$ is unreachable and thus $\thlfp{R}\, X @ P \relsemi X$ reduces to $P$. Whether or not $P$ terminates, productivity
is the criterion needed (see Definition~\ref{def:productive}), and from this we can prove the following theorem.

\begin{theorem} \label{thm:productive}
  If $P$ is productive then the function $\lambda X @ P \relsemi X$ is guarded. \isalink{https://github.com/isabelle-utp/utp-main/blob/90ec1d65d63e91a69fbfeeafe69bd7d67f753a47/theories/rea_designs/utp_rdes_guarded.thy\#L156}
\end{theorem}

\noindent So, by Theorem~\ref{thm:uniqfp} we can map the weakest fixed-point to the strongest fixed-point. This brings
us to Kleene's fixed-point theorem, which allows us to calculate an iterative construction for the strongest fixed-point
of a continuous function.

\begin{theorem}[Kleene's fixed-point theorem] \label{thm:kleene} If $F$ is a continuous function then the strongest
  fixed-point can be calculated by iteration: $$\nu F ~~=~~ \bigsqcap_{i\in\nat} F^i(\false)$$
\end{theorem}

\noindent Kleene's fixed-point theorem often employs Scott-continuity as its antecedent, which is based on complete
partial orders rather than complete lattices. We employ our stronger notion of continuity in Theorem~\ref{thm:kleene},
since we do have a complete lattice in our setting. Theorem~\ref{thm:kleene} allows us to calculate the fixed-point by
iterating $F$, starting from $\false$, which is the top ($\top$) of the relational lattice. Now, for our simplified pattern
$\nu X @ P \relsemi X$ we automatically have continuity since relational composition is continuous, that is
$$(\bigsqcap i @ P(i)) \relsemi Q = (\bigsqcap i @ P(i) \relsemi Q)$$ is a theorem of relational calculus (see Theorem~\ref{thm:rellaws}). Combining this
property with Theorem~\ref{thm:rdes-comp-pow}, which includes the calculation for a power construction, we are now in
the position to substantiate Theorem~\ref{thm:rdes-rec} for calculating iterative reactive contracts. We give the
calculational proof below explicitly since it uses and illustrates several of our results.

\begin{proof}[Proof of Theorem~\ref{thm:rdes-rec}] \isalink{https://github.com/isabelle-utp/utp-main/blob/90ec1d65d63e91a69fbfeeafe69bd7d67f753a47/theories/rea_designs/utp_rdes_guarded.thy\#L279}
\begin{align*}
  &\quad~ \thlfp{R}\, X @ \rc{P}{Q}{R} \relsemi X  \\[1ex]
  &= \mu X @ \rc{P}{Q}{R} \relsemi \healthy{SRD}(X) & \text{[Theorem~\ref{thm:contthm-lfp}]} \\[1ex]
  &= \nu X @ \rc{P}{Q}{R} \relsemi \healthy{SRD}(X) & \text{[Theorems~\ref{thm:uniqfp} and \ref{thm:productive}]}\\[1ex]
  &= \bigsqcap_{i \in \nat} (\lambda X @ \rc{P}{Q}{R} \relsemi \healthy{SRD}(X))^i(\false) & \text{[Theorem \ref{thm:kleene}]} \\[1ex]
  &= \bigsqcap_{i \in \nat} (\lambda X @ \rc{P}{Q}{R} \relsemi \healthy{SRD}(X))^{i+1}(\false)&  \text{[Unfold: $f^0(\false) = \false)$]} \\[1ex]
  &= \bigsqcap_{i \in \nat} (\rc{P}{Q}{R}^{i+1} \relsemi \healthy{SRD}(\false)) & \text{[Induction on $i$]} \\[1ex]
  &= \bigsqcap_{i \in \nat} (\rc{P}{Q}{R}^{i+1} \relsemi \Miracle) & \text{[Theorem~\ref{thm:utp-lattice-top}]} \\[1ex]
  &= \bigsqcap_{i \in \nat} \rc{\bigwedge_{i \le n} \left(R^i \wpR P\right)}{\bigvee_{i \le n} R^i \relsemi Q}{R^{n+1}} \relsemi \Miracle & \text{[Theorem~\ref{thm:rdes-comp-pow}]} \\[1ex]
  &= \bigsqcap_{i \in \nat} \rc{\bigwedge_{i \le n} \left(R^i \wpR P\right)}{\bigvee_{i \le n} R^i \relsemi Q}{\false} & \text{[Theorem~\ref{law:RD9}]} &  \\[1ex]
  &= \rc{\bigwedge_{i \in \nat} \left(R^i \wpR P\right)}{\bigvee_{i \in \nat} R^i \relsemi Q}{\false} 
  & \left[\begin{array}{r} \text{Theorem~\ref{thm:rdes-comp-inf} and} \\ \text{relational calculus}
    \end{array}\right] & \qedhere
\end{align*}
\end{proof}

\noindent This proof demonstrates the necessity of a large corpus of theorems we have proved from the UTP theories and
reactive designs in order to reason about recursion. We now have a complete constructive approach for calculating the
reactive contract for a variety of recursive specifications. Extension to deal with mutual recursion is laborious, but
not challenging.

%% file: sec/mech.tex
In this section, we give an overview of the mechanisation work in Isabelle/UTP~\cite{Foster16a,Zeyda16}, which has
supported the development of our theory of reactive contracts, and has allowed us to turn these theories into a
prototype verification tool.

The theory hierarchy described in Sections~\ref{sec:rrel}, \ref{sec:contracts}, and \ref{sec:grd}, has been developed
almost entirely mechanically. Whilst we have planned the main high-level theorems using pen and paper prior to
mechanisation, all the low-level lemmas have been proved automatically using the proof tactics we have developed for
UTP. This has significantly helped our theory engineering. We have been able to progress at a speedy pace, with proof
support aiding the exploration of theory properties and discovery of missing or implicit assumptions in our
theorems. This has had a positive impact on the precision of our theories by preventing formalisation gaps in proofs,
and improves our overall confidence in the correctness of our theorems. This is particularly the case due to the LCF
architecture of Isabelle/HOL, which ensures that our theories are all sound with respect to the axioms of higher order
logic. The most difficult technical challenge of this work has been mechanising the parallel-by-merge
construct~\cite{Hoare&98}, which requires that we perform non-trivial manipulations on alphabets; for details please see
the accompanying mechanised artefacts in Section~\ref{sec:parallel}.

We have also used our theory hierarchy to produce a refinement-based verification tactic for reactive contracts.  This,
in particular, allows us to harness a number of existing automated proof tactics and strategies~\cite{Blanchette2011} in
verifying reactive programs, such as \textsf{auto}, which automates logical deduction, and \textsf{sledgehammer}, which
integrates external automated theorems and SMT solvers.

Each of the rules for contract calculation given in Section~\ref{sec:contracts} have been proved correct with respect to
our underlying operator definitions of relational calculus. These theorems then act as inputs to the proof tactic, which
proceeds by first calculating the contract for a reactive program, and then attempting to automatically prove that it
refines a contractual specification. A proof goal about a potentially complex reactive program can be effectively
reduced into three simpler reactive relations characterising assumptions, intermediate behaviours, and final behaviours,
which can be more readily discharged and so favours greater automation. This is the approach we have prototyped in our
verification example in Section~\ref{sec:mondex}.

The initial particular problem to be solved is bridging the gap between UTP predicates and HOL predicates, which are
distinct types, albeit strongly related. In particular, the majority of predicate and relational operators of
Isabelle/UTP are defined, via the \textsf{lifting} package~\cite{Huffman13}, in terms of operators of HOL. The major
difference is the addition of lenses~\cite{Foster16a}, which accounts for the state variables and alphabet that is not
explicitly present in HOL predicates. Nevertheless, because UTP predicates are effectively an enrichment of HOL
predicates, laws and tactics applicable to the former often can be adapted to the latter. This, therefore, prevents the
need to ``reinvent the wheel'' when conducting automated reasoning in Isabelle/UTP.

Our approach to reasoning about alphabetised predicates and relations therefore extends the approach we have first
described in~\cite{Foster14}. The basic approach is two step: firstly apply the \textsf{transfer} tactic of the
\textsf{lifting} package~\cite{Huffman13} to interpret a UTP predicate as a HOL predicate, and secondly apply
Isabelle/HOL's built in automated reasoning tactics. We have extended this approach to provide explicit support for
lenses, such that each UTP variable can be collapsed and rewritten to a HOL variable with a similar name. Effectively
this means that reasoning about UTP predicates and relations can often be entirely reduced to reasoning about HOL
predicates, which improves proof automation.

In particular, for relations, proof reduces to a satisfiability problem, since sequential composition boils down to
existential quantification, substitution, and conjunction. We provide the \textsf{rel-auto} tactic, which uses this
approach, combined with HOL's \textsf{auto} tactic, to deal with conjectures in the predicate and relational
calculus. It can discharge the majority of core Isabelle/UTP laws automatically, though in some cases a subsequent call
to automated theorem provers using \textsf{sledgehammer} is also necessary.

Reasoning about predicates and relations in this way is very powerful and efficient, but it can be further optimised by
making use of specific patterns imposed by a UTP theory. Each reactive contract is written in terms of the $\Rs$
healthiness condition, and the design turnstile, which adds a large relational overhead particularly when dealing with
large composite definitions. As we indicated, this machinery can be collapsed using the theorems of Sections
\ref{sec:contracts} and \ref{sec:grd}, so that proof is performed over the pre-, peri-, and postconditions, which are
all simply reactive relations. Therefore, we have created a series of additional tactics to automate this process, and
in particular the tactics \textsf{rdes-refine} and \textsf{rdes-eq} that employ the calculational laws, and
Theorems~\ref{thm:rdesrefine} and \ref{thm:rdesequiv}, to prove (or refute) refinement conjectures and equality
conjectures.

In particular, this tactic can be used to the automate contract verification problem highlighted in
Section~\ref{sec:contracts}, which takes the form of $$\rc{P_1}{P_2}{P_3} \sqsubseteq Q$$ and states that some reactive
program $Q$, defined using operators of the object language, satisfies the given contractual specification.

This \textsf{rdes-refine} and \textsf{rdes-eq} tactics are generic, in that they can be applied to prove conjectures of
any reactive contract based language, so long as denotational definitions and theorems are available that show how to
expand the operators to contracts. The proof process has three stages:

\begin{enumerate}
\item \textbf{Calculate and Simplify Reactive Design Contracts}. The operators of the object language, for example
  \Circus, are first expanded into their contractual specifications. For a large reactive program this could result in
  three very large predicates, which are addressed separately in the next stage. This stage requires that appropriate
  calculation theorems have been proved and are already available. The algebraic laws of reactive designs, for example
  Theorem~\ref{thm:rdes-comp}, are applied to convert a composite reactive contract into a monolithic one of the
  syntactic form $\rc{P_1}{P_2}{P_3}$. The application of these theorems and the associated simplifications can be
  performed separately by a tactic called \textsf{rdes-simp}.
\item \textbf{Apply the Refinement Theorem}. Theorem~\ref{thm:rdesrefine} is applied to break the refinement conjecture
  into three proof obligations about the pre-, peri-, and postconditions. As usual, we require that the precondition is
  weakened, and the peri- and postcondition is strengthened.
\item \textbf{Apply Relational Transfer and Reasoning Tactics}. The three proof obligations are purely relational, do
  not refer to observational variables $ok$ and $wait$, and thus are relatively lightweight in nature. Thus we can now
  apply \textsf{rel-auto}, and potentially other Isabelle/HOL tactics, to try and complete the proof. Of course, we can
  also at this point try to generate a counterexample, which can be traced back to a fault in the pre-, peri-, or
  postcondition.
\end{enumerate}

\noindent We have found this approach extremely useful for automated reasoning, and it has been used to validate all the
example calculations, in particular those related to the cash-card example\footnote{See
  \url{https://github.com/isabelle-utp/utp-main/blob/master/tutorial/utp_csp_mini_mondex.thy} for our complete
  mechanised theory} from Section~\ref{sec:mondex}. The tactic is also very general, in that any language whose
operators can be expressed using contracts, automatically receives proof support. In this sense, whilst so far we have
applied it only to \Circus, it can also be applied to a variety of related
languages~\cite{Wei2013,Woodcock14,He94,Zhan2008}, and therefore Isabelle/UTP constitutes a generic toolkit for
verification tools.

%% file: sec/related.tex
Meyer~\cite{Meyer92} coined the term ``design-by-contract'' when arguing the need for precisely specified assertions of
a program's behaviour to ensure reliability of component-based software systems. These assertions come in three forms:
preconditions, postconditions, and invariants, which are used to annotate methods and attributes within classes in his
object-oriented programming language, Eiffel~\cite{Meyer88}.

Meyer's work has its foundation in that of Floyd~\cite{Floyd67} and Hoare~\cite{Hoare69} on proving correctness of
programs in terms of the Hoare triple $\hoaretriple{p}{Q}{r}$. This asserts that if program $Q$ is started in a state
satisfying predicate $p$ then, provided $Q$ terminates, the final state satisfies predicate $r$. Effectively, the Hoare
triple sets up a contract on state variables for $Q$ that mandates a particular relationship between inputs and outputs.

Refinement calculi~\cite{Back1998,Mor1996,Morris1987} promote such pre/postcondition specifications to statements of the
language itself. A programming language is extended with an abstract specification statement, such as $w\text{:}[pre,
post]$~\cite{Mor1996}, where $pre$ is the precondition, $post$ the postcondition, and $w$ the ``frame'' -- that is, the
set of variables that are allowed to change. A specification statement can be transformed into a block of executable
code through a series of correctness-preserving refinements that divide the specification into constituent parts, and
eventually introduce atomic commands. Morgan~\cite{Mor1996} and Back~\cite{Back1998} both refer to such specification
statements as ``contracts'' between the specifier and the implementer, who is allowed to weaken the precondition and
strengthen the postcondition. This makes the implementation more deterministic and thus predictable, whilst fulfilling
the original contract that specifies it.

Hoare logic and refinement calculus can be unified either algebraically, using Kleene
algebras~\cite{Armstrong2015,Gomes2016}, or else denotationally~\cite{Cavalcanti04,He2006,Oliveira&09}, using the UTP
theory of designs~\cite{Hoare&98,Cavalcanti04}, which is foundational for our reactive contracts. The unifying nature of
UTP is achieved through encoding all conceivable programs, not as bespoke mathematical structures, but as elements
of the alphabetised predicate calculus.

Benveniste et al.~\cite{Benveniste2007,Benvenuti2008} provide a formal contract framework that has been very influential
in the area of contract-based design~\cite{Vincentelli2012,Bauer2012,Cimatti2015}. They give a denotational framework
that describes contracts as pairs of predicates $C \defs (A, G)$ where $A$ is the assumption, and $G$ is the
guarantee. $A$ and $G$ both characterise the set of program variable traces that prescribe valid behaviours of the
context and given component, respectively. They define a notion of refinement $C_1 \preceq C_2$ called ``dominance''
--- named to distinguish it from refinement to implementation, but otherwise technically the same --- that requires that
the assumption be weakened and the guarantee be strengthened. They also define an algebra for contracts including
parallel composition $C_1 \parallel C_2$, disjunction $C_1 \sqcap C_2$, and conjunction $C_1 \sqcup C_2$.  Dominance
corresponds to the universal notion of refinement in UTP: $P \refinedby Q$.

We note that Benveniste's theory has a striking resemblance to the theory of designs, although designs do not explicitly
account for traces. Designs effectively encode the pair of assumption-guarantee predicates into a single relation via
the $ok$ and $ok'$ variables. Thus, many of the laws in \cite{Benveniste2007} have very similar design theorems, in
particular those for the lattice operators. This implies that derivatives of the UTP design theory, including ours,
rightly fit into their general contract framework~\cite{Benveniste2007,Benvenuti2008}, and the resulting work
stream~\cite{Vincentelli2012,Bauer2012,Cimatti2015}.

Works building on Benveniste's framework~\cite{Cimatti2012,Cimatti2015} allow step-wise refinement, and verification
based on temporal logic properties. Previous works~\cite{vonKarger1995,vonKarger2000} indicate that such temporal logics
can be readily characterised in our domain. Our extensible semantic model is more expressive than a purely trace-based
model, and we are limited to neither trace nor failures-divergences refinement. Our contract model also explicitly
distinguishes terminating behaviours, and thus supports sequential as well as parallel composition. We also handle state
variables, so that they need not be modelled as a sequence of updates in the trace.

On its own, the UTP theory of designs only accounts for imperative behaviour and thus various specialisations have been
made for other paradigms. Notable is the theory of reactive designs~\cite{Cavalcanti&06,Oliveira&09}, which adds program
histories represented by traces, alongside an account of non-terminating behaviours. Effectively, the theory of reactive
designs combines the theory of designs with the theory of reactive processes. This integration of stateful and reactive
behaviour enabled Oliveira et al.~\cite{Oliveira2005-PHD,Oliveira&09} to give a UTP semantics to the \Circus language,
which combines CSP~\cite{Hoare85,Roscoe2005}, the Z notation~\cite{Spivey89}, and guarded command
language~\cite{Dijkstra75}.

Reactive designs arise as a direct consequence of a result due to Hoare and He~\cite[Theorem 8.2.2]{Hoare&98}, and was
named as such by Cavalcanti and Woodcock~\cite{Cavalcanti&06} who highlighted a normal form:
$\healthy{R}(\design{P_1}{P_2})$, which is intuitively a design made reactive with assumption $P_1$ and guarantee
$P_2$.. They can be used as a semantic domain for reactive languages with assumptions and guarantees, and is the
foundational idea behind our theory. The innovation of our work is to subdivide $P_2$ into two parts, for intermediate
and final observations, and prove a large set of calculational theorems that support automated verification. Our
reactive designs theory also implements a contract framework with a generic trace model, based on trace
algebra~\cite{Foster17a}, similar to \cite{Benveniste2007,Cimatti2015}. However, we embed the traces directly into the
design predicate through \healthy{R1} and $\healthy{R2}_c$. This means that the relational calculus operators are
directly applicable without redefinition, and substantiates the unifying nature of our theory.

Our work also overcomes a number of technical limitations in the existing reactive designs
theory~\cite{Oliveira2005-PHD,Oliveira&09}. Firstly, we explicitly characterise internal state with an observational
variable, and require that the after state is invisible in intermediate states. This enables correct interaction with
stateful behaviours for operators like external choice~\cite{BGW09}. Secondly, we relax a restriction of
\cite{Oliveira2005-PHD,Oliveira&09} on reactive design assumptions, which in their previous work cannot mention traces
and thus are unsuitable to represent reactive assumptions.

Butterfield et al.~\cite{BGW09} study state visibility in the theory of reactive designs, noting that the observation of
program state is miraculous while waiting for an external choice to be resolved for $\healthy{R3}$ predicates. This is
because in the original account~\cite{Hoare&98} $\healthy{R3}$ leaves intermediate valuations of state variables
observable, however, CSP's external choice operator conjoins the intermediate observations since, until an event occurs,
all such behaviours have potential to resolve. Therefore, inconsistent intermediate values for state variables lead to
miraculous behaviour since, for example, $(x := 1 \land x := 2) ~=~ \false$. Thus, in \cite{BGW09}, \healthy{R3} is
replaced by a variant called $\healthy{R3}_h$, which causes internal state updates to be abstracted in intermediate
observations. They note that this curtails the ability to interrupt an action and pass its intermediate state valuation
onto the interrupting action~\cite{BGW09,McEwan06}. Nevertheless, we adopt $\healthy{R3}_h$ in our work, as it leads to
a simpler handling of state and has several advantages for automated reasoning. We also note that it is not difficult to
adapt our mechanisation to a version of reactive contracts based on $\healthy{R3}$ instead.

Reactive designs have been further extended to account for discrete time in the UTP theories for the languages
\CircusTime~\cite{Sherif2010,Wei2013} and \CML~\cite{Woodcock14,Canham15}. They augment the trace with time events,
albeit in slightly different ways. Moreover, a theory similar to that of reactive designs has been used to give a
denotational semantics to a hybrid version of CSP~\cite{He94} and the dynamic systems modelling language
Modelica~\cite{Foster16b}. In that work, each continuous variable has a continuous-time trace. Our reactive contract
theory unifies these discrete and hybrid theories. Canham and Woodcock~\cite{Canham15} explicitly distinguish
non-terminating and terminating behaviours in the postcondition, resulting in a triple of the form
$\healthy{R}(\design{P}{Q_1 \wcond Q_2})$ that we generalise and adopt. A key result of our trace algebraic foundations
is that all the aforementioned extensions can be unified by our UTP theory.

Another extension of the theory of designs to account for contractual behaviour is found in
rCOS~\cite{He2005,He2006,Zhan2008}, a refinement calculus and UTP theory for object-oriented and component systems. A
contract in rCOS is expressed as a quadruple consisting of (1) the component's interface, that is, the types of its
attributes and methods; (2) an initialiser or constructor for the component; (3) a collection of specifications for the
methods; and (4) the valid traces of calls to methods that can be made. The initialiser and methods are specified using
a form of UTP design also called a ``reactive design''~\cite{He2005}. A contract's dynamic behaviour is given a
semantics in the style of CSP's failures-divergences, with method names as events.

In rCOS~\cite{He2005,He2006,Zhan2008}, the valid traces are used to protect the method calls that can deadlock if their
preconditions are violated, and thus act as a form of environmental assumption. Their notion of reactive design uses
only the healthiness condition \healthy{R3}. Moreover, their contract notion is not embedded in UTP's relational
calculus, but are explicit quadruples. This means they cannot reuse operators like sequential composition and
nondeterministic choice, hampering composition with other UTP theories. Nevertheless, their language, being based on a
CSP-style failures-divergence semantics augmented with stateful operations with design semantics, can be readily
embedded into \Circus, and hence our reactive design theory.

The rely-guarantee technique of Jones et al.~\cite{Jones81,Jones2003,Hayes2016a} provides contract-based reasoning for
programs with shared variable concurrency. A rely-guarantee quintuple $\{p,r\}c\{g,q\}$ states that, assuming the
initial state satisfies predicate $p$, and each atomic step of the environment satisfies relation $r$ -- the rely
condition -- then program $c$ terminates in a state satisfying $q$ and guarantees $g$ at each of its atomic steps. The
rely condition specifies how much interference the process must tolerate from its siblings, and in return, the guarantee
condition puts a limit on the interference the process can cause through its manipulation of the shared variables. A
co-existence proof obligation requires that each parallel program is guaranteeing enough for what the other relies
upon. Another proof obligation is that the postcondition is implemented by the steps taken by the two programs and their
environment.

It might seem that the rely-guarantee technique, being based on shared variables, is not directly applicable to
languages like \Circus, which is based on a communication model with no variable sharing. However, our contracts are
based on an abstract notion of trace that can encompass shared variable updates as events. In our reactive design
contracts, the precondition can refer to both the trace and initial value of state variables, and therefore encompasses
both a precondition $p$ and rely condition $r$. Our pericondition broadly corresponds to a guarantee condition $g$,
since it states what must be true at each intermediate step. However, our postcondition can also refer to the completed
trace at the point of termination, along with the final valuation of state variables, which may afford us additional
expressivity.

Comparison with rely-guarantee becomes easier if we consider a more abstract algebraic level. Hayes et al. have explored
algebras to characterise rely-guarantee reasoning~\cite{Hayes2016a,Hayes2016b,Colvin2017}, which they call Concurrent
Refinement Algebra (CRA). Like the theory of designs, CRA is based on a refinement lattice with nondeterministic
choice. It also provides operators for sequential composition, parallel composition, and iteration. In \cite{Hayes2016b},
they identify a subalgebra of CRA that corresponds to atomic steps. In addition to characterising atomic state updates,
the algebra is also used to characterise CSP-style events, which more readily compares with our work. 

Hayes' semantics for CRA is based on Aczel traces~\cite{Colvin2017}, which explicitly distinguish passive environmental
events and active program events. The former is used to encode that a process is willing to allow a particular event to
occur, whilst the latter expresses active engagement. The resulting semantic model is very different to the standard
failures-divergence model of CSP and \Circus, in that an event prefix $a \then P$ explicitly prescribes that all other
possible events can occur before $a$ occurs in the trace, by using environment events. This allows a unified notion of
parallel composition, which simply reduces the possible environmental and program events.

Their model does not currently account for external choice, though this should be representable with the choice being
resolved only through a program event, and not environment events. Nevertheless, the relation to the standard model of
CSP remains to be seen and a soundness result would be necessary to use their semantic model as a drop-in
replacement. On the other hand, our semantic model is sufficiently general to encompass finite Aczel traces, though for
this paper we focus on (though do not mandate) the standard failures-divergences model in examples.

%% file: sec/conclusions.tex
We have developed a comprehensive and generalised UTP theory of reactive designs that can be applied to contract-based
modelling and automated verification of sequential and parallel reactive systems. Our language of contracts allows to
both compose specifications for reactive programs, and also give denotational semantics to reactive languages in the
programs-as-predicates approach~\cite{Hehner93}. At a specification level, our contracts allow us to restrict
permissible behaviours of the environment using preconditions, and specify possible behaviours in intermediate and final
observations using peri- and postconditions. This is supported by a theory of reactive relations
(Section~\ref{sec:rrel}), with which we are able to specify predicates and relations that refer to both state variable
valuations and the trace of interactions. This theory provides an abstract algebraic account of traces, and therefore
supports both discrete event sequences and also piecewise continuous functions. The theory is therefore applicable to
both discrete time and hybrid systems.

From our theory, we derived a set of calculational laws in Section~\ref{sec:contracts} for reactive contract composition
using operators like sequential composition, non-deterministic choice, recursion, and parallel composition. We also
proved theorems that demonstrate refinement and equivalence between two contracts, and use these to develop a prototype
procedure for automated verification of reactive programs by calculation. This procedure was applied to a small
cash-card verification example in Section~\ref{sec:mondex} based in the \Circus~\cite{Woodcock2001-Circus} language. The
underlying healthiness conditions for our UTP theory have also been formulated in Section~\ref{sec:grd}, and a number of
key theorems for tail recursion and parallel composition have been expounded. Our contract theory is practically
supported by a mechanisation in Isabelle/UTP, which provides both confidence in the proven theorems, and also automated
verification facilities through a number of proof tactics. The vast majority of definitions and theorems from
Section~\ref{sec:rrel} onwards in this paper are novel; they exceptions are Theorems~\ref{thm:uniqfp}, \ref{thm:kleene},
and \ref{def:par-by-merge}.

There are a number of important areas for future work. In the area of assume-guarantee reasoning, there are a number of
works that have much in common with our framework~\cite{Benveniste2007,Benvenuti2008,Cimatti2015}, particularly at the
level of the fundamental UTP theory of designs. In the future, it would be interesting to fully explore these links, for
example, through formalisation of their theories in UTP and formation of suitable Galois connections~\cite{Hoare&98},
especially with recent work on hybrid system contracts~\cite{Cimatti2015}. Such a unification could also consider
formalisation of the Aczel trace model~\cite{Colvin2017}, a semantic model that facilitates event
frames~\cite{Hayes2016b}, which distinguish events \emph{engaged in} from those simply permitted by a system constituent. This
in turn could permit a much simpler, but no less expressive, definition of parallel composition provided by CSP and
\Circus and thus improve support for compositional reasoning with contracts.

Also with respect to parallel composition, the weakest-rely-condition calculus ($\wppR{M}$) will be further developed
to capture interference between concurrent processes, and establish concrete links with rely-guarantee
algebra~\cite{Hayes2016a} and concurrent Kleene algebra~\cite{Hayes2016b}. There is also need for a mechanised
refinement calculus, to support step-wise development of reactive programs. This should, in particular, support laws for
parallel composition that facilitate compositional reasoning through distributing invariants, which can draw on previous
work with \Circus~\cite{Oliveira2005-PHD}.

With respect to recursion, there is also potential future work. In this paper we have considered calculation of
contracts for tail recursive programs utilising Hoare and He's guardedness theorem~\cite{Hoare&98}, and Kleene's fixed
point theorem~\cite{Lassez82}. Whilst this encompasses a significant number of reactive programs, there is also
potential for more general recursion schemes. For example, the \Circus action
$$\mu X @ (a \then P \relsemi X) \extchoice b \then Q$$ which expresses a recursive body $P$ guarded by event $a$, with
escape event $b$ leading to $Q$, is clearly productive and guarded, but does not fit our tail recursive
pattern. Therefore, more general patterns could be identified and calculational laws proved utilising Theorems
\ref{thm:rdes-comp-pow} and \ref{thm:uniqfp}, and reusing our notion of productivity. One can potentially go
further and consider non-continuous recursion schemes to which Kleene's fixed-point theorem is not applicable anymore,
and can build on a further theorem by Hoare and He for calculating general recursive UTP designs~\cite[theorem 3.1.6,
page 81]{Hoare&98}.

In a different axis, we have yet to fully explore the use of our theory in the context of continuous-time trace
models. Previously, we have used such a model to give a UTP semantics to the dynamical systems modelling language
Modelica~\cite{Foster16b}, but in a setting without formal contracts. Thus, in the future, we could consider a contract
language for concurrent hybrid systems using the infrastructure present in this paper, together with a number of
specialisations. We already have some preliminary results on modelling block-based control law diagrams using hybrid
reactive contracts, and hope to publish this work in the near future.

Finally, more experience is needed in the use of our automated verification tool in order to explore applications from a
variety of domains. We would thus like to complete our mechanisation of \Circus, and apply it to verify a more
substantial case study. Clearly, this would rely upon the availability of parallel reasoning facilities highlighted
above. Moreover, verification tools for other UTP-based contract languages like \CML will be developed. We can then
start to gather results about the efficiency of our tool, and begin to work on improvements towards and applicability to
large and realistic industrial case studies.